\pdfoutput=1
\documentclass[acmsmall,screen]{acmart}
 \setcopyright{none}
 \makeatletter
\renewcommand\@formatdoi[1]{\ignorespaces}
\renewcommand\@setauthorsaddresses{\ignorespaces}
\makeatother

\title{Online Detection of Effectively Callback Free Objects with Applications to Smart Contracts}         

\author{Shelly Grossman}
\affiliation{
  \department{Computer Science}              
  \institution{Tel Aviv University}            
  \country{Israel}
}

\author{Ittai Abraham}
\affiliation{
  \institution{VMware Research}            
  \country{USA}
}

\author{Guy Golan-Gueta}
\affiliation{
  \institution{VMware Research}            
  \country{USA}
}

\author{Yan Michalevsky}
\affiliation{
  \institution{Stanford University}            
  \country{USA}
}

\author{Noam Rinetzky}
\affiliation{
  \department{Computer Science}              
  \institution{Tel Aviv University}            
  \country{Israel}
}

\author{Mooly Sagiv}
\affiliation{
  \department{Computer Science}              
  \institution{Tel Aviv University}            
  \country{Israel}
}
\affiliation{
  \institution{VMware Research}            
  \country{USA}
}

\author{Yoni Zohar}
\affiliation{
  \department{Computer Science}              
  \institution{Tel Aviv University}            
  \country{Israel}
}

\renewcommand{\shortauthors}{Grossman, Abraham, Golan-Gueta, Michalevsky, Rinetzky, Sagiv, and Zohar}



\usepackage{booktabs}
\usepackage{subcaption} 
\usepackage{alltt}
\usepackage{amsmath}
\usepackage{amsthm}
\usepackage{stmaryrd}
\usepackage{amsfonts}
\usepackage[all]{xy}
\usepackage{listings}
\usepackage{mathtools}
\usepackage{paralist}
\usepackage[T1]{fontenc}
\usepackage{extarrows}
\usepackage[noend,noline]{algorithm2e}
\usepackage{arydshln}
\usepackage{centernot}
\usepackage{multirow}
\usepackage{multicol}
\usepackage[nodayofweek]{datetime}
\usepackage{todonotes}
\usepackage{subcaption}
\usepackage{pgf,tikz}
\usepackage{cleveref} 

\newcommand{\ignore}[1]{}

\newcommand{\shellyold}[1]{}
\newcommand{\red}[1]{\color{red}#1\color{black}}

\newcommand{\code}[1]{\texttt{#1}}

\newcommand{\tbd}[1]{}
\newcommand{\changed}[1]{{#1}}

\makeatletter
\newcommand{\@@eqLikeDef}[3]{%
    \ensuremath{\overset{\mathclap{\text{\scalebox{#1}{#2}}}}{#3}}%
}
\newcommand*{\@eqLikeDef}[2]{
    \mathchoice
        {
            \@@eqLikeDef{0.7}{#1}{#2}
        }
        {
            \@@eqLikeDef{0.7}{#1}{#2}
        }
        {
            \@@eqLikeDef{0.6}{#1}{#2}
        }
        {
            \@@eqLikeDef{0.5}{#1}{#2}
        }
}

\newcommand*{\eqdef}{\@eqLikeDef{def}{=}}
\makeatother

\newcommand{\PDA}{\text{PDA}}
\newcommand{\B}[1]{\langle #1 \rangle}
\newcommand{\ST}[1]{(#1)}
\newcommand{\SET}[1]{\{#1\}}
\newcommand{\partialto}{\rightharpoonup}

\newcommand{\dom}{\operatorname{dom}}

\newcommand{\res}{\code{res}}
\newcommand{\ret}{\code{ret}}
\newcommand{\inarg}{\code{arg}}
\newcommand{\formal}[1]{\operatorname{formal}(o)}

\newcommand{\inctx}{\vdash}

\newcommand{\PCmd}{\syn{PCmd}}
\newcommand{\Cmd}{\syn{Cmd}}
\newcommand{\Contract}{\syn{Contract}}

\newcommand{\SAssign}[2]{#1 \coloneqq  #2}

\newcommand{\SAssert}[1]{\code{assert}(#1)}
\newcommand{\SSkip}[0]{\code{skip}}

\newcommand{\SCallC}[3]{{#1 \coloneqq #2(#3)}}

\newcommand{\SReturn}[0]{\code{return}}
\newcommand{\SReturnV}{\code{return}}
\newcommand{\SEnterV}{\code{enter}}

\newcommand{\cdone}{\code{done}}

\newcommand{\PLname}{\ensuremath{\mathbf{SMAC}}}

\newcommand{\syn}[1]{\mathsf{#1}}
\newcommand{\Var}[0]{\syn{LVar}}

\newcommand{\semtrue}[0]{\it{true}}
\newcommand{\semfalse}[0]{\it{false}}

\newcommand{\SelStack}{\Stk}
\newcommand{\SelObject}{\objectid}
\newcommand{\SelCmd}{C}
\newcommand{\SelStore}{\store}
\newcommand{\SelEnv}{\reg}

\newcommand{\sem}[1]{\mathbf{#1}}
\newcommand{\Reg}[0]{\sem{Reg}}

\newcommand{\Stack}{\sem{Stack}}
\newcommand{\Store}{\sem{Store}}
\newcommand{\Frame}{\sem{Frame}}
\newcommand{\State}{\sem{State}}

\DeclareFontFamily{U} {MnSymbolA}{}
\DeclareFontShape{U}{MnSymbolA}{m}{n}{
  <-6> MnSymbolA5
  <6-7> MnSymbolA6
  <7-8> MnSymbolA7
  <8-9> MnSymbolA8
  <9-10> MnSymbolA9
  <10-12> MnSymbolA10
  <12-> MnSymbolA12}{}
\DeclareFontShape{U}{MnSymbolA}{b}{n}{
  <-6> MnSymbolA-Bold5
  <6-7> MnSymbolA-Bold6
  <7-8> MnSymbolA-Bold7
  <8-9> MnSymbolA-Bold8
  <9-10> MnSymbolA-Bold9
  <10-12> MnSymbolA-Bold10
  <12-> MnSymbolA-Bold12}{}

\DeclareFontFamily{U}  {MnSymbolB}{}
\DeclareFontShape{U}{MnSymbolB}{m}{n}{
    <-6>  MnSymbolB5
   <6-7>  MnSymbolB6
   <7-8>  MnSymbolB7
   <8-9>  MnSymbolB8
   <9-10> MnSymbolB9
  <10-12> MnSymbolB10
  <12->   MnSymbolB12}{}
\DeclareFontShape{U}{MnSymbolB}{b}{n}{
    <-6>  MnSymbolB-Bold5
   <6-7>  MnSymbolB-Bold6
   <7-8>  MnSymbolB-Bold7
   <8-9>  MnSymbolB-Bold8
   <9-10> MnSymbolB-Bold9
  <10-12> MnSymbolB-Bold10
  <12->   MnSymbolB-Bold12}{}

\DeclareSymbolFont{MnSyA} {U} {MnSymbolA}{m}{n}
\DeclareSymbolFont{MnSyB}         {U}  {MnSymbolB}{m}{n}

\DeclareMathSymbol{\rcirclearrowright}{\mathrel}{MnSyA}{248}
\DeclareMathSymbol{\nrcirclearrowright}{\mathrel}{MnSyB}{248}
\DeclareMathSymbol{\leftrightarrows}{\mathrel}{MnSyA}{152}
\DeclareMathSymbol{\nleftrightarrows}{\mathrel}{MnSyB}{152}

\newcommand{\ocode}{\kappa}
\newcommand{\ObjVar}{\sem{Heap}}
\newcommand{\oovar}{\operatorname{\psi}}

\newcommand{\objectid}{o}
\newcommand{\cmd}[0]{c}

\newcommand{\Stk}[0]{\Gamma}
\newcommand{\stk}[0]{\gamma}
\newcommand{\topstk}{\operatorname{top}}
\newcommand{\reg}[0]{\rho}

\newcommand{\defostatestack}[5]{\langle \frame{#1}{#2}{#3}\cdot #4, #5 \rangle}
\newcommand{\defostate}[4]{\defostatestack{#1}{#2}{#3}{\Stk}{#4}}

\newcommand{\ostatedefault}{\langle \Stk,\store \rangle}

\newcommand{\store}{\sigma}

\newcommand{\Tr}{\mathit{Tr}}

\newcommand{\execdepth}{\it{Depth}}

\newcommand{\inv}{\it{I}}

\newcommand{\sstate}{s}
\renewcommand{\frame}[3]{\ST{#1,\,#2,\,#3}}

\newcommand{\src}{\operatorname{src}}
\newcommand{\trg}{\operatorname{trg}}

\newcommand{\exec}[0]{\pi}

\newcommand{\Conflict}{\it{Conflict}}

\newcommand{\cep}{\varphi}
\newcommand{\transition}{\tau}
\newcommand{\trid}{\iota}

\newcommand{\event}{e}
\newcommand{\trace}{T}
\newcommand{\action}{a}

\newcommand{\subexec}{\sqsubseteq}

\newcommand{\fseq}[0]{\simeq_{\scriptscriptstyle FS}}
\newcommand{\fseqo}[1]{\fseq^{#1}}
\newcommand{\ceq}[0]{\simeq_{C}}

\newcommand{\ECF}{\textsc{ECF}}
\newcommand{\SECF}[1][]{\textsc{sECF}#1}
\newcommand{\DECF}[1][]{\textsc{dECF}#1}
\newcommand{\DECFo}[2][]{\textsc{dECF}^{#2}\!\!#1}
\newcommand{\FSECF}{\textsc{ECF}\subFS}
\newcommand{\CECF}{\textsc{ECF}\subC}

\newcommand{\subC}{\textsubscript{\textsc{c}}}
\newcommand{\subFS}{\textsubscript{\textsc{fs}}}

\newcommand{\segment}[0]{t}

\newcommand{\readset}[1]{R\ifthenelse{\equal{#1}{}}{}{(#1)}}
\newcommand{\writeset}[1]{W\ifthenelse{\equal{#1}{}}{}{(#1)}}
\newcommand{\depth}[1]{D\ifthenelse{\equal{#1}{}}{}{(#1)}}
\newcommand{\depthNoArg}{D}
\newcommand{\indexInCall}[1]{i\ifthenelse{\equal{#1}{}}{}{(#1)}}
\newcommand{\indexInExec}[1]{Idx\ifthenelse{\equal{#1}{}}{}{(#1)}}
\newcommand{\indexInExecNoArg}{Idx}
\newcommand{\prefix}[1]{\it{prefix}(#1)}
\newcommand{\suffix}[1]{\it{suffix}(#1)}
\newcommand{\Commute}[2]{#1\ensuremath{\leftrightarrows}#2}
\newcommand{\nCommute}[2]{#1\ensuremath{\nleftrightarrows}#2}
\newcommand{\hb}[2]{#1\ensuremath{\prec_{\it{Inv}}}#2}

\begin{document}



\begin{abstract}
Callbacks are essential in many programming environments, but drastically
complicate program understanding and reasoning because they allow to mutate
object's local states by external objects in unexpected fashions, thus breaking modularity.
The famous DAO bug in the cryptocurrency framework \emph{Ethereum},  employed callbacks to steal \$150M.
We define the notion of Effectively Callback Free (ECF) objects in order to allow
callbacks without preventing modular reasoning.

An object is ECF in a given execution trace if there exists an
equivalent execution trace without callbacks to this object. 
An object is ECF if it is ECF in every possible execution trace. 
We study the decidability of dynamically checking ECF in a given execution trace
and statically checking if an object is ECF.
We also show that dynamically checking ECF in Ethereum is feasible and can be done online.
By running  the history of all execution traces in Ethereum, 
we were able to verify that virtually all existing contract executions,
 excluding these of the DAO or of contracts with similar known vulnerabilities, are ECF.
Finally, we show that ECF, whether it is verified dynamically or statically,
 enables modular reasoning about objects with encapsulated state.
\end{abstract}

\maketitle



\section{Introduction}
The theme of this paper is enabling modular reasoning about the
correctness of objects with encapsulated state. 
This is inspired by platforms like Ethereum~\cite{Eth:EthYellowPaperGavWood} 
that facilitate execution of \emph{Smart Contracts}~\cite{Szabo1997formalizing} on top of a blockchain-based distributed ledger~\cite{nakamoto2008bitcoin}.
A key property in Ethereum Smart Contracts is the lack of global
mutable shared state, in contrast to common standard programming
environments such as C and Java. 
A smart contract is analogous to an object with encapsulated state.

However, the Ethereum blockchain, and many other dynamic environments, implement event-driven
programming using callbacks.
These callbacks are necessary for functionality, but can compromise security.
For example, the famous bug in the DAO contract exploited callbacks to steal \$150M~\cite{Dao}.

Indeed, callbacks may break modularity which is essential for good programming style and extendibility.
In the context of Blockchain, modularity is even more important since contracts are contributed by different sources,
some of which may be malicious. 
Accordingly, the bug in the DAO allowed an adversarially crafted contract to mutate the DAO's state by calling back to it.

The DAO contract, that implemented a crowd-funding platform, was attacked by a `callback loop-hole' (to be precisely described below). 
This attack, the recovery from which required a controversial hard-fork of the blockchain,\footnote{A hard fork can be thought of as taking an agreed history of transactions, and manually change it.} 
exhibits a vulnerability that is peculiar to 
decentralized consensus systems, like Ethereum: 
	in such systems, a buggy contract cannot be updated or fixed 
		(except for extreme measures like hard-forking), 
	which makes validation and verification of smart contracts 
	of even greater importance for this application.

\paragraph{Effectively Callback-Free Objects}
We identify a natural generic correctness criteria for objects 
which enables modular reasoning in environments with local-only mutable states, 
and expect most correct objects to satisfy this requirement.
Informally, if an object $o$ calls another object $o'$, and the execution of $o'$ calls $o$ again, this second call to $o$ is defined as a callback.
The main idea is to allow callbacks in $o$ only when they cannot affect the serial
non-interruptible behavior of $o$.
Thus, such callbacks can be considered harmless and do not affect the set of local reachable states
of the object $o$.
In particular, the behavior of such objects is independent of the client environments and of other objects.
It is possible to reproduce all behaviors of the object using a most general client and without analyzing external objects.

We say that an execution is
\textbf{Dynamically Effectively Callback Free} (\DECF{}) when there exists ``an
equivalent'' execution without callbacks which
starts in the same state and reaches the same final state.
By \emph{equivalent}, we refer to the behavior of a particular object
as an external observer may perceive.
We say that an object is
\textbf{Statically Effectively Callback Free} (\SECF{}) when all its possible
executions are dynamically ECF.
We do not distinguish between dynamic and static ECF when the context is clear.
Both definitions are useful.
Dynamic ECF in particular is applicable to the blockchain environment,
since static ECF is undecidable in the general case.
We ran experiments on Ethereum, proving that checking dynamic ECF is inexpensive,
and thus can be done efficiently in-vivo. 
This, combined with Ethereum's built-in rollback feature,
would have allowed to prevent the DAO bug from occurring,
without invalidating legitimate executions.
(In fact, we found just one such legitimate non-ECF contract,  discussed in \Cref{Sec:Evaluation}).


We show that the
vulnerable DAO contract is non-ECF while no non-ECF executions are detected
after applying the suggested corrections to it.
Notice that the ECF notion is similar to the notion of atomic
transactions in concurrent systems.
Indeed, despite the fact that contract languages do not usually
support concurrency, modularity and callbacks require similar kind
of reasoning.

The ECF property's usefulness is not limited to bug-finding; 
once ECF is established, it can be served to simplify reasoning on the object in isolation of other objects:
We show that the set of reachable local states in ECF objects can
be determined without considering the code of other objects and
thus enable modular reasoning. This modular reasoning can be
performed automatically using abstract interpretation e.g., as
suggested in \citet{LogozzoCLSS09} or by using deductive verification
which is supported by Dafny~\cite{Leino2010}.
We demonstrate this by
verifying an interesting invariant of the DAO contract. (See \Cref{Se:Overview}). 

\paragraph{Online Detection of ECF executions}
A na{\"i}ve detection of \DECF{} may be costly because of the need to
enumerate subexecution traces.
Therefore, we develop an effective polynomial online algorithm for checking if an
execution is ECF.
The main idea is to detect conflicting memory accesses and utilize commutativity in an effective manner.
We integrated the algorithm into the \emph{Ethereum Virtual Machine (EVM)}~\cite{Eth:EthYellowPaperGavWood}.
We ran the algorithm on all executions kept in the Ethereum blockchain until 23 June 2017, and demonstrate that:
(i)~the vulnerable DAO contract and other buggy contracts are non-ECF.
(ii)~very few correct contracts are non-ECF,
(iii)~callbacks are not esoteric and are used in many contracts, 
and (iv)~the runtime overhead of our implementation is negligible and
thus can be integrated as an online check.

This online detection can thus be used to prevent incidents like the theft from the DAO 
at the cost of slightly more restricted form of programming.

As far as we are aware, our tool is also the most precise and effective tool 
for finding such vulnerable behaviors due to callbacks.
We compared it to the Oyente tool~\cite{RW:LuuCCS16,Eth:OyenteOnline}, by giving it both ECF and non-ECF contracts  based on the DAO object (\Cref{Fi:DaoContract}). 
We found that it has false positives, as it detects a `reentrancy bug' (the common name of the DAO vulnerability in the blockchain community) for any one of the fixes that render our example contract ECF.

\paragraph{Decidability of \SECF{} for objects}
We also consider the problem of checking \SECF{} algorithmically.
Obviously, since modern contract languages, such as Solidity~\cite{Eth:Solidity}, are
Turing complete, checking if a contract is ECF
is undecidable.

We show that checking that a contract is \SECF{} in a language with finite local states
is decidable.
This is interesting since many contracts only use small local states or maps with uniform data
independent accesses.
Technically, this result is non-trivial since the nesting of contract calls is unbounded,
and since ECF requires reasoning about permutations of nested invocations.
The reason for the decidability is that non-ECF executions which occur in high depth of nesting
must also occur in depth $2$.

\paragraph{Main Results}
Our results can be summarized as follows:
\begin{enumerate}
\item We define a general safety property, called ECF, for objects (\SECF{}) and executions (\DECF{}).
Our definition is inspired by the Blockchain environment but it may also be useful for other environments with encapsulated states, such as Microservices.
\item We show that objects with encapsulated data, under the assumption that they satisfy ECF, can be verified using modular reasoning in a sound manner.
\item A stronger notion of ECF, based on conflict-equivalence, 
enabling efficient verification of \DECF{} in real-life environments, 
and for which \SECF{} is decidable for programs with finite state and unbounded stack.
\item A polynomial time and space algorithm for online checking of \DECF{} 
and prototype implementation of it as a dynamic monitor of \DECF{}, built on top of an Ethereum client.
\item Evaluation of the algorithm on the entire history of the Ethereum blockchain (both main and `Classic' forks, see \Cref{Sec:Evaluation}).
The monitor detects true bad executions (the infamous DAO and others) as non-ECF, and has near-zero false positives.
Based on this result, it can be inferred that, in practice, most non-ECF executions correspond to bad executions.
We also show that our monitor has a very small runtime overhead.
By retroactively running the \DECF{} monitor on the available history, 
we were able to prove its effectiveness in preventing the exploitation of the vulnerability in the DAO,
and even more importantly, 
the feasibility of leveraging it in other applications, e.g., simplifying 
modular contract verification.{} 
\end{enumerate}

\section{Overview}
\label{Se:Overview}
\newcommand{\icredit}{\textit{credit}}
\newcommand{\ibalance}{\textit{balance}}

This section provides some necessary background and an informal overview of our approach.

\subsection{The DAO Bug}

\Cref{Fi:DaoContract} shows pseudocode illustrating the vulnerability in the DAO.\footnote{\emph{DAO} is acronym for \emph{decentralized autonomous organization}, and its purpose is to facilitate voting on proposals and on investments by the owners of the DAO.}
The contract stores a \code{credit} for each object, as well as the current \code{balance}.\footnote{In
programming languages like Solidity, balance is a predefined field of every contract, maintained by the runtime system. We write it explicitly for clarity.}
The \code{credit} represents individual investments per object.
To align with the Ethereum terminology, the unit of currency represented by \code{credit} and \code{balance} is called \emph{ether}.
The contract maintains a representation invariant, where the sum of the credits equals to the current balance,
i.e.,
\begin{equation}
\Sigma_{o\in\dom(\icredit)} \icredit[o] = \ibalance
\label{eq:DaoInvariant}
\end{equation}
The contract offers two methods for manipulating states:
\texttt{deposit} for depositing money and \texttt{withdrawAll} for
withdrawing all available funds of a specific object.

\begin{figure}
\begin{minipage}{1in}
\begin{footnotesize}
\begin{alltt}
\begin{tabbing}
XX\=XX\=XX\=XX\=XX\=XX\=XX\=XX\=XX\=XX\=XX\=\kill
Object DAO \+\\
Map<Object, int> credit\\
int balance\\
\textbf{Invariant} \textit{(sum o: credit[o]) = balance}\\
\\
Method withdrawAll(Object o)         Method deposit(Object o, int amount)  \+\\
2: if (oCredit > 0)                    6: credit[o] += amount   \+\\
// 2.5: credit[o] = 0                7: this.balance += amount\\
3: this.balance -= oCredit \\
4: o.pay(oCredit) \\
5: credit[o] = 0 \-\\
\end{tabbing}
\end{alltt}
\end{footnotesize}
\end{minipage}
\vspace*{-5mm}
\caption{\label{Fi:DaoContract}%
A contract illustrating the DAO bug.
The representation invariant may be violated by callbacks from malicious contracts.
Line $2.5$ fixes the bug.}
\end{figure}

\begin{figure}  
\centering
  \begin{subfigure}[t]{2.5in}
    \centering
    \begin{footnotesize}
    \begin{alltt}
\begin{tabbing}
Object GoodClient \\
  Object Dao, int balance \\
  Method init(Object dao) \\
    1: this.Dao = dao \\
  Method pay(int profit)  \\
    2: this.balance += profit \\
  Method depositCredit(\=Object dao, int amount) \\
    3: Dao.deposit(this, amount) \\
  Method getCredit(Object dao) \\
    4: Dao.withdrawAll(this) 
\end{tabbing}
    \end{alltt}
    \end{footnotesize}
\caption{\label{Fi:DaoGoodClient}%
An innocent client using the \code{DAO} object without violating its representation
invariant.}
  \end{subfigure}
  \quad
  \begin{subfigure}[t]{2.5in}
    \centering
\begin{footnotesize}
\begin{alltt}
\begin{tabbing}
Object Attacker \\
  Object Dao, bool stop, int balance\\
  Method init(Object dao)   \\
    1: Dao = dao \\
    2: stop = false  \\
  Method pay(int profit)  \\
    3: this.balance += profit\\
    4: if (!stop)    \\
       5: stop = true   \\
       6: Dao.withdrawAll(this) \\
    7: stop = false 
\end{tabbing}
\end{alltt}
\end{footnotesize}
    \caption{\label{Fi:DaoAttack}%
A snippet of an \code{Attacker} object. It is stealing money from the \code{DAO} object by violating its representation
invariant.}
  \end{subfigure}
  \caption{An innocent and a malicious client using the \code{DAO} object}\label{Fi:DaoClients}
\end{figure}

\newcommand{\statefig}[1]{\hbox{$\begin{array}{l}#1\end{array}$}}

\begin{figure}
\[
\resizebox{\textwidth}{!}{%
\xymatrix@C14pt{
*+[F=]\statefig{D.c[G]=100\\D.c[A]=100\\D.b=200\\A.b=0\\A.s=\semfalse}
                \ar[r]^-{w \{}_-1 &*+[F]\statefig{\text{read }D.c[A]=100\\D.b=100}
                \ar[r]^-{p \{ }_-2&*+[F]\statefig{A.b=100\\A.s=\semtrue}
                \ar[r]^-{w \{ }_-3&*+[F]\statefig{\text{read }D.c[A]=100\\D.b=0}
                \ar[r]^-{p \{ }_-4 
&*+[F]\statefig{A.b=200\\A.s=\semtrue}
    \ar[dlll]^-{p \} }_-5 \\
    &*[F]\statefig{D.c[A]=0} 
    \ar[r]^-{w \} }_-6&*+[F]\statefig{A.s=\semfalse}
    \ar[r]^-{p \} }_-7&*+[F]\statefig{D.c[A]=0}
    \ar[r]^-{w \} }_-8&*+[F=]\statefig{D.c[G]=100\\D.c[A]=0\\D.b=0\\A.b=200\\A.s=\semfalse}
}
}
\]
\caption{\label{Fi:StealTrace}%
A trace of calls illustrating an attack on the \code{DAO}.
Nodes are labeled by local changed states and edges are labeled by actions and by the corresponding order in the original trace.
$D$ denotes the \code{DAO}, $G$ denotes a \texttt{GoodClient} and $A$ is an \texttt{Attacker} object.
$w$ denotes the \texttt{withdrawAll} operation and $p$ denotes the \texttt{pay} operation.
$b$ is a shorthand for \texttt{balance}, $c$ is a shorthand for \texttt{credit}, and $s$ is a shorthand for \code{stop}. 
}
\end{figure}
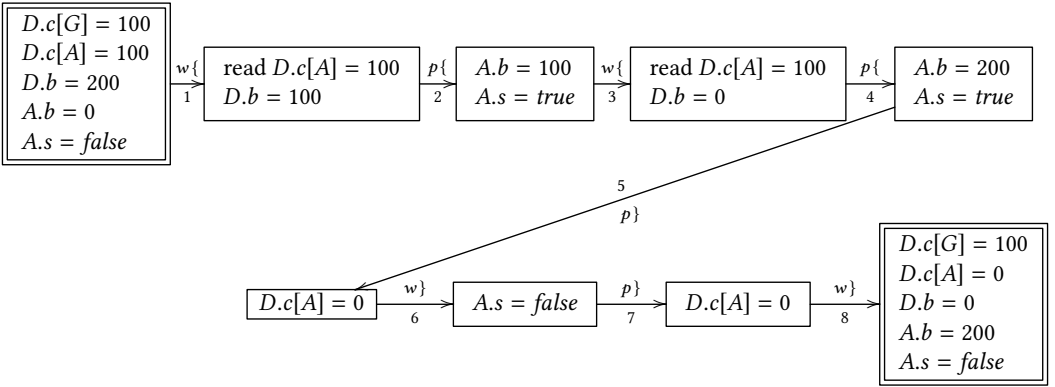

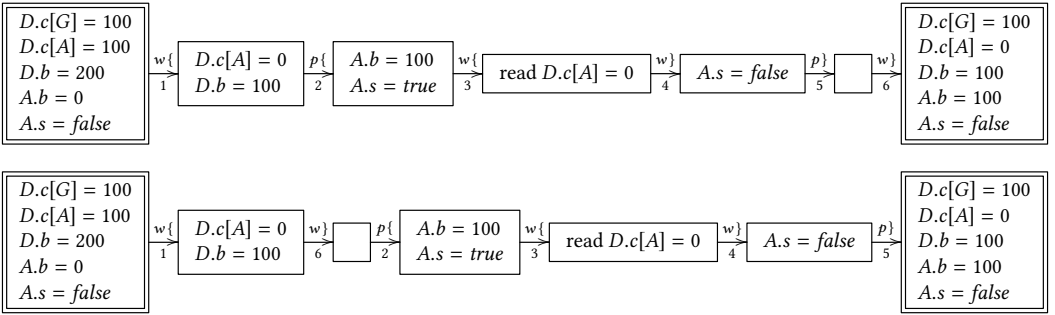
\begin{figure}
\[
\resizebox{\textwidth}{!}{%
\xymatrix@C14pt{
*+[F=]\statefig{D.c[G]=100\\D.c[A]=100\\D.b=200\\A.b=0\\A.s=\semfalse}
                \ar[r]^-{w \{}_-1 &*+[F]\statefig{D.c[A]=0\\D.b=100}
                \ar[r]^-{p \{ }_-2 &*+[F]\statefig{A.b=100\\A.s=\semtrue}
                \ar[r]^-{w \{ }_-3 &*+[F]\statefig{\text{read }D.c[A]=0}
                \ar[r]^-{w \} }_-4 &*+[F]\statefig{A.s=\semfalse}
    \ar[r]^-{p \} }_-5 
    & *+[F]\statefig{\,}
    \ar[r]^-{w \} }_-6 &*+[F=]\statefig{D.c[G]=100\\D.c[A]=0\\D.b=100\\A.b=100\\A.s=\semfalse}
}
}
\]

\[
\resizebox{\textwidth}{!}{%
\xymatrix@C14pt{
*+[F=]\statefig{D.c[G]=100\\D.c[A]=100\\D.b=200\\A.b=0\\A.s=\semfalse}
                \ar[r]^-{w \{}_-1 &*+[F]\statefig{D.c[A]=0\\D.b=100}
                \ar[r]^-{w \} }_-6 &*+[F]\statefig{\,}
                \ar[r]^-{p \{ }_-2 &*+[F]\statefig{A.b=100\\A.s=\semtrue}
                \ar[r]^-{w \{ }_-3 &*+[F]\statefig{\text{read }D.c[A]=0}
                \ar[r]^-{w \} }_-4 &*+[F]\statefig{A.s=\semfalse}
    \ar[r]^-{p \}}_-5 
    & *+[F=]\statefig{D.c[G]=100\\D.c[A]=0\\D.b=100\\A.b=100\\A.s=\semfalse}
}
}
\]
\caption{\label{Fi:FailedStealInECF}%
Two traces of calls illustrating the original and callback-free versions of a failed attack on the ECF version of the \code{DAO}.
}
\end{figure}

\Cref{Fi:DaoGoodClient} shows a simple client illustrating the expected usages of the \code{DAO} object.
\Cref{Fi:DaoAttack} shows a simple attack on the \code{DAO} object.
The code callbacks to the \code{DAO} method \texttt{withdrawAll} to steal money.
\Cref{Fi:StealTrace} depicts a concrete trace of attacking the \code{DAO} assuming that the \code{DAO}'s initial balance is $200$ ether.
We reached that state after a \code{GoodClient} object and an \code{Attacker} object deposited each $100$ ether.
In the first call to \code{withdrawAll}, the attacker will get the amount he invested originally in the attack ($100$ ether).
The \code{DAO} then calls to the \code{Attacker} object's \code{pay} method, which increases the attacker's balance by $100$ ether, and calls \code{withdrawAll} again.
The \code{pay} method is designed to call \code{withdrawAll} at most once in a trace by updating the \code{stop} variable, and avoid infinite recursion.\footnote{For clarity, we avoid technical discussion of the semantics of executions and exceptions in Ethereum/Solidity, to allow us to focus on the ECF property.}
The code of \code{withdrawAll} in the second run will transfer an additional $100$ ether from the \code{DAO} object to the attacker.
In the end of the trace, the \code{DAO} was depleted of its funds completely, and the attacker managed to illegitimately receive the funds that belonged to \code{GoodClient}.

\subsection{Effectively Callback Free Contracts}

In principle, semi-automatic program verification and abstract interpretation can be used to verify the absence of malicious attacks
like the one in the \texttt{Attacker} object.
However, this requires reasoning about the whole code.\footnote{In the case of Ethereum, it is in fact impossible to reason about the whole code, as new contracts can be added at any time, and these contracts could interact with the contract being checked.}
This paper advocates a different solution by exploring modularity.
The idea is to require stronger conditions from the contracts which prevent the need to reason about other objects at all.

Specifically, we define the notion of \textbf{effectively
callback free} (ECF) objects.
Our definition is inspired by Blockchain contracts but is applicable to enforce modularity in other environments with local states. 

We say that an execution of an object with an initial state $s_0$ and final state $s$ is
\textbf{Dynamically Effectively Callback Free} (\DECF{}) when there exists ``an
equivalent'' execution of the contract without callbacks which
starts in the same initial state $s_0$ and
reaches the final state $s$.
We say that an object is
\textbf{Statically Effectively Callback Free} (\SECF{}) when all its possible
executions are effectively callback free.

The \texttt{DAO} object is not ECF. For example, the trace depicted in \Cref{Fi:StealTrace} cannot be reproduced without callbacks to reach the same state.
In contrast, the fix to the \texttt{DAO} object by uncommenting line $2.5$ and deleting line $5$ makes the contract ECF.
This contract is now ECF since all its traces can be reordered to avoid callbacks.
For example, \Cref{Fi:FailedStealInECF} shows a trace of an attempt to perform an attack similar to the attack in \Cref{Fi:StealTrace} and its corresponding reordering that avoids callbacks. 
Note, that in the reordered trace, \code{withdrawAll} did not execute line $4$.
Omitting calls is allowed for the sake of proving an execution is ECF, 
as our goal is to be able to reproduce, assuming there are no callbacks, the same behaviors that are feasible with callbacks.

\subsection{Online ECF Detection}
It is possible to check in a na{\"i}ve way that an execution is ECF by recording the trace and checking the ECF property at the end of the execution, by enumerating all possible permutations.
However, this is costly both in space and in time, since the number of permutations grows exponentially with the size of the trace.
In particular, it is hard to see if such a solution can be integrated into a virtual machine.

In order to obtain a feasible online algorithm, we check a stronger requirement than ECF, which is inspired by conflict serializability of database transactions.
The main idea is to explore commutativity of operations for efficient online checking of a correctness condition which guarantees that the callback-free trace results in the same state as the original trace. 

Consider a trace $\exec$ with potential callbacks and a reordered trace $\exec'$ which does not include callbacks.
$\exec'$ is not necessarily feasible, unless we permit to ignore external calls by objects and force the clients to perform these calls instead.
We say that $\exec$ and $\exec'$ are 
 conflict equivalent 
if every pair of conflicting read/write operations in $\exec$ appear in the same order in $\exec'$.
Operations conflict when they are not commutative.
Commutativity is mechanically checked by comparing the read and the write sets of operations, and forbidding intersection of read/write conflicts. 
For example, in \Cref{Fi:StealTrace}, the read operation of \code{D.c[A]} in the \code{withdrawAll} action labeled 1 (lines 1-3) does not commute with the write operation of \code{D.c[A]} in the \code{withdrawAll} action labeled 5 (line 5).
However, the top trace in \Cref{Fi:FailedStealInECF}, depicting a trace of the ECF version of \code{DAO}, the operations in the \code{withdrawAll} action labeled 3 (lines 1-2) commute with the operations in the \code{withdrawAll} action labeled 6, which has an empty read and write sets (no code was executed since line 5 was deleted).
The information regarding commutativity of different subtraces is used to build a constraint graph on the ordering of object invocations. 
When this constraint graph contains no cycles, it is possible to perform topological sort to find a concrete callback-free trace.
A full description of the algorithm and its complexity is available in \Cref{Sec:Dynamic}.

We integrated this algorithm into the EVM (the \emph{Ethereum Virtual Machine}) and applied it to all available executions in the blockchain.
The results are summarized in \Cref{Sec:Evaluation}.
They indicate that the vast majority of non-ECF executions come from erroneous contracts. 
They also indicate that the runtime overhead of our instrumentation is neglectable.  
From these encouraging results, we concluded that if the ECF check was part of the Ethereum protocol, 
	it could have prevented the vulnerability in the \code{DAO} from being exploited.
	Its clearly beneficiary for an environment like Ethereum, which handles sensitive financial transactions, 
	and in which code is virtually impossible to upgrade.

\subsection{Deciding ECF Contracts}

We also investigated the possibility to verify at compile-time that a contract is ECF (the \SECF{}  property).
In general, this is undecidable, since languages such as \emph{Solidity}, a high-level front-end to EVM bytecode, are Turing-complete.
However, we show that for contracts with finite local states, checking ECF is decidable.
This result is non-trivial as the model allows for an unbounded stack length.
The decision procedure devised provides insight on additional techniques for checking ECF in practice.

\subsection{Verifying Properties of ECF Contracts}
In this paper, we show that reasoning about ECF contracts can be performed
in a modular fashion.
The local reachable states of an ECF contract are only affected by the 
code of the contract, 
and cannot be changed by external contracts.

This is useful for program verification and program analysis, treating 
external calls 
as 
non-deterministic operations that may return an arbitrary value, but cannot change the local state.
We utilized this property using Dafny~\cite{Leino2010}, to verify correctness of the revised \code{DAO} object from \Cref{Fi:DaoContract} (including line 2.5, excluding line 5).
When doing so, we ignored the call in line 4, because the return value was not used. 
We provide a deeper discussion on verifying this example using Dafny in \Cref{Sec:Modularity}.




\subsection{Summary of the Rest of This Paper}
The paper is organized as follows:
In \Cref{Sec:Preliminaries} we formally present the syntax and semantics of our programming language for contracts,
called $\PLname$. 
Notions of equivalence are presented in \Cref{Se:EqExec}. 
The ECF property and its different `flavors' (dynamic vs. static, and for different equivalence notions) are presented in \Cref{Sec:CorrectnessConditions}.
We discuss decidability results for ECF in \Cref{Sec:Decidability}.
\Cref{Sec:Modularity} shows the application of the ECF property to achieve modular object-level analysis.
The algorithm for online verification of \DECF{} is given in full in \Cref{Sec:Dynamic}.
We discuss our experimental results obtained by running the algorithm on the Ethereum blockchain in \Cref{Sec:Evaluation}.
Related work is provided in \Cref{Sec:RelatedWork} and we conclude in \Cref{Sec:Conclusion}.



\newcommand{\cid}{k}
\newcommand{\kid}{k}
\newcommand{\oid}{o}
\newcommand{\CID}{\mathsf{Cnt}}
\newcommand{\FID}{\mathsf{Fld}}
\newcommand{\VID}{\mathsf{Var}}
\newcommand{\statedepth}{\operatorname{Depth}}
\newcommand{\minexecdepth}{\operatorname{minDepth}}
\newcommand{\maxexecdepth}{\operatorname{maxDepth}}
\newcommand{\execstates}{\operatorname{States}}
\newcommand{\exectrs}{\operatorname{Transitions}}
\newcommand{\tr}{\transition}

\section{Programming Language}\label{Se:PL}\label{Sec:Preliminaries}

We formalize our results for \(\PLname\), a simple imperative object-based programming language with
pass-by-value parameters with integer-typed local variables and data members (fields).
For simplicity, and without loss of generality, every method has a single formal parameter named $\inarg$ and 
returns a value by assigning it to a designated variable $\ret$.
Even though we present our theoretical development for contracts in $\PLname$, for readability we use
a Java-like notation in our examples, which can be easily desugared.

\begin{figure}
\[
\begin{array}{lclcl}
c & \in & \PCmd & \eqdef & \SAssign{x}{e} \mid \SAssign{F}{x} \mid \SAssign{x}{F}  \mid \SAssert{b}  \mid  
             \SCallC{x}{\oid}{e} \mid \SSkip \mid \SEnterV \mid \SReturnV
\\
C & \in & \Cmd & \eqdef & c \mid C \, ; \, C \mid \code{if}\,\,\,b\,\,\,\code{then}\,\,C \,\,\code{else}\,\, C \mid \code{while}\,b\,\code{do}\,\,C 
\\
K & \in & \Contract & \eqdef & \cid \colon \overline{f} \, \,\, \, \SEnterV\, \,\code{var}\,\,\overline{x}\,\,\, \, C \,\,\,\, \SReturnV
\end{array}
\]
\caption{Syntax. 
}
\label{fig:Syntax}
\end{figure}

\subsection{Syntax}\label{Sec:Syntax}

\Cref{fig:Syntax} defines the syntax of $\PLname$.
We assume infinite syntactic domains of 
$\cid\in\CID$, $f \in \FID$, and $x\in\VID$ 
\emph{contract identifiers}, \emph{field names}, and \emph{variable identifiers}, respectively.
A contract $K$ is identified by a (unique) \emph{contract identifier} $\cid$,    
and contains a sequence of \emph{field} definitions $\overline{f}$ 
and  a (single) nameless \emph{contract method}.
The contract method is comprised of a sequence of \emph{local variable} definitions $\overline{x}$ and a command $C\in \Cmd$.
$C$ may be a \emph{primitive command} $c \in \PCmd$  
or a \emph{compound command}, i.e., a sequential composition of commands, a conditional, or a loop.
A primitive command $c\in\PCmd$ 
may be either an assignment of an expression $e$ to a local variable $x$ ($\SAssign{x}{e})$,
an assignment of the value of a local variable $x$ to a field $F$ ($\SAssign{F}{x}$),
an assignment of the value of a field $F$ to a local variable $x$ ($\SAssign{x}{F}$),
an \emph{assert} command ($\SAssert{b}$),
 a call to a contract method  with a single argument $e$, 
		keeping the returned value in a local variable $x$ ($\SCallC{x}{\oid}{e}$),
		or a $\SSkip$ command.
Each contract has a single method, thus methods are not named, and may be colloquially referred to using the name of their contract.

Without loss of generatility, we assume that no two contracts contain a field with the same name.
In the following, we use the terms `contract' and `object' interchangeably.

\begin{figure}
$
\centering
\begin{array}{ll}
\begin{array}{l@{\,}l@{\,}l@{\,}c@{\,}ll}
\reg     & \in   & \Reg       & =     & \VID \partialto_{\mathit{fin}} \mathbb{Z} & \text{Local states} \\ 
\oovar   & \in   & \ObjVar    & =     & \FID \rightarrow_{\mathit{fin}} \mathbb{Z} & \text{Object states}\\ 
\stk     & \in   & \Frame    & =     & \CID \times \Cmd \times \Reg & \text{Frames}\\
\end{array}
&
\begin{array}{l@{\,}l@{\,}l@{\,}c@{\,}ll}
\Stk     & \in   & \Stack    & =     & \overline{\Frame}& \text{Stacks}\\
\store   & \in   & \Store    & =     & \CID \partialto_{\mathit{fin}} \ObjVar & \text{Stores} \\
\sstate   & \in   & \State    & =     & \Stack \times   \Store & \text{States}\\
\end{array}
\end{array}
$
\caption{Semantic domains.\label{fig:Semantic-domains}}
\end{figure}

\subsection{Semantics}\label{SubSec:Semantics}

$\PLname$ has a rather mundane stack-based operational semantics, which handle method calls using a \emph{stack} of activation records (\emph{frames}), and uses a \emph{store}
to record the values stored in object fields.
We refer to a state in which the stack is empty as a \emph{quiescent} state and to a non-quiescent state  as an \emph{active} state.
Once the execution reaches a \emph{quiescent} state, \emph{any} object method may start running.
We refer the reader's attention to three important points:
\begin{inparaenum}[(i)]
\item contract states are encapsulated: A contract $\oid$ can only access its own fields,
\item local variables are private to their invocation, and
\item \label{I:deterministic} once a contract method is invoked, the semantics is deterministic.
\end{inparaenum}
We denote the  \emph{code context} (\emph{context} for short) which provides the code of every called object by $\ocode\in\CID\partialto_{fin}\Cmd$, i.e., $\ocode(\objectid)$ denotes the   code of   object $\objectid$.

\paragraph{States}
\Cref{fig:Semantic-domains} defines the semantic domains.
A \emph{state} $\sstate=\ostatedefault$ is a pair comprised of a (possibly empty) \emph{stack} of \emph{frames} $\Stk\in\Stack$ 
and a \emph{store} $\store \in \Store$, 
denoted by $\SelStack(\sstate)=\Stk$ and $\SelStore(\sstate)=\store$, respectively.
The \emph{depth} of a state $\sstate$, denoted by $\statedepth(\sstate)$, is the number of elements in its stack, i.e., $\statedepth(\sstate)=|\SelStack(\sstate)|$.

We denote the \emph{top} of the stack in an active state $\sstate=\ostatedefault$ by $\topstk(\sstate)=\Stk(1)$.
Intuitively, $\topstk(\sstate)$ contains the \emph{local state} of the \emph{active} (i.e., currently executing) contract method, while the other frames record the locals states of \emph{pending} calls to contract methods.
A \emph{frame} $\stk=\frame{\objectid}{\cmd}{\reg}$ records the \emph{local state} of 
(a call to the contract method of) an object. 
Formally, $\stk$ is a triple comprised of an \emph{object identifier}, denoted by $\SelObject(\stk)=\objectid$, a command, denoted by $\SelCmd(\stk)=\cmd$, which the method needs to execute, 
and a \emph{local environment} $\reg\in\Reg$, denoted by $\SelEnv(\stk)=\reg$, which assigns values to the invocation's local variables.
A \emph{store} $\store\in\Store$ is a mapping from a finite number of object identifiers to their   \emph{object state}.

\begin{figure}
\centering
$
\begin{array}{rclll}
\B{\epsilon,\store} & \Rightarrow & \B{\frame{\oid}{\ocode(o)}{[\inarg \mapsto n]}, \, \store} 
\\
\B{\frame{\oid}{\SReturnV}{\reg},\, \store} & \Rightarrow &   \B{\epsilon,\store} & 
\\[1ex]
\B{\frame{\oid}{\SCallC{x}{\oid'}{e}}{\reg} \cdot \stk,\, \store} & \Rightarrow & \B{\ST{o',\ocode(o'),[\inarg \mapsto \reg(e)]} \cdot \frame{\oid}{\SAssign{x}{\res}}{\reg}\cdot \stk, \store} 
\\
\B{\frame{\oid'}{\SReturn}{\reg'} \cdot \frame{\oid}{\cmd}{\reg} \cdot \Stk, \store} 
& \Rightarrow &   
\B{\frame{\oid}{\cmd}{\reg[\res \mapsto \reg'(\ret)]} \cdot \Stk, \store} 
\end{array}
$
\caption{Operational semantics with a \emph{context} $\ocode$.
$\reg$ is naturally extended for expressions over variables in $\Var$. 
We denote $o \in \dom(\store), n \in\mathbb{N}$.
\label{Fig:CallOS}
}
\end{figure}

\paragraph{Transition relations}
We formalize the semantics of our programming language using a \emph{transition relation}.
A \emph{transition} is a pair $\transition=(\sstate,\sstate') \in \Tr \subseteq \State \times \State$ 
comprised of 
a \emph{source state} $\sstate$, denoted by $\src(\transition)$,
and \emph{target} state $\sstate'$, denoted by  $\trg(\transition)$.
For clarity, we sometimes write a $\transition=(\sstate,\sstate')$ as $\sstate \Rightarrow \sstate'$.
We denote the active object of the transition by $\objectid(\transition)=\objectid(\topstk(\src(\transition)))$, or $\oid_{\it{main}}$ if it starts in a quiescent state.
We denote by $\cmd(\transition)\in\PCmd$ 
the primitive command that justifies the transition. 

The meaning of primitive and compound commands is standard, and thus omitted.
We mention that primitive commands can only use local variables taken  from the top stack frame, 
and that only the fields of the \emph{active object} can be accessed.

\Cref{Fig:CallOS} defines meaning of  method calls and returns. 
When an object  $\objectid$ is called from a quiescent state, a new stack frame is pushed to the currently empty stack. 
The frame determines that the active object is $\objectid$, the command executing is the code $\ocode(\objectid)$ of $\objectid$, and the local environment for the invocation is the assignment of the value of $n$ to $\inarg$.
The last command in $\ocode(\objectid)$ is always a $\SReturn$, after which the frame is popped, leading to a quiescent state.
When a call $\SCallC{x}{\objectid'}{e}$ is made from an active state, a new stack frame is pushed as in the previous case.
We note that the local environment is initialized by assigning to $\inarg$ the value of $e$ in the local environment belonging to the caller, $\reg(e)$.
To handle retrieval of the return value from the callee, the command in the caller is modified to assign to $x$ the value of a specially designated variable $\res$. 
When the callee invocation of $\objectid'$ finishes, the command in the top frame is $\SReturn$ and we let $\reg'$ denote the local environment of the callee. 
The control transfers back to the caller object $\objectid$, and the value of $\res$ is set to be the value of $\ret$ in $\reg'$. 
The assigned value of $\res$ is then automatically assigned to $x$, as determined by the operational semantics of the call.
The primitive command associated with a call is $\SEnterV$, and with a return is $\SReturnV$. 

\paragraph{Executions}
An \emph{execution} $\exec = \exec(1) \ldots \exec(|\exec|)$ is a finite sequence of transitions coming from $\Tr$.
An execution $\exec$ is \emph{well-formed} if the target state of every transition is the source state of the following one, i.e.,
$\forall i\in\{2..|\exec|\}.\,\trg(\exec(i-1))=\src(\exec(i))$.
For clarity, we sometimes write an execution $\exec$ as $\exec =  \sstate_1 \Rightarrow \sstate_2 \Rightarrow \cdots \sstate_n$.
We use $\ocode \inctx \exec$ to denote that an execution $\exec$ takes objects' code from context $\ocode$. We omit the context when no confusion is likely.

We say that a transition $\transition$ \emph{appears in} a $\exec$, denoted by $\transition \in \exec$, if $\exec = \_ \cdot \transition \cdot \_ $.
We say that a state $\sstate$ \emph{appears in} a $\exec$, denoted by $\sstate \in \exec$, if there is a transition $\transition \in \exec$ such that
$\sstate \in \{\src(\transition),\trg(\transition)\}$. 
We denote the sets of  transitions and states that appear in an execution $\exec$ by $\execstates(\exec)$ and $\exectrs(\exec)$, respectively. 
An \emph{execution} $\exec'$ is a \emph{subexecution} of an execution $\exec$,   denoted by $\exec'\subexec\exec$, if it is a subsequence of $\exec$. 

We denote the \emph{first} and \emph{last} states of a non-empty execution $\exec$ by $\src(\exec)=\src(\exec(1))$ and $\trg(\exec)=\trg(\exec(|\exec|))$.
We say that $\exec=\tr\exec'\tr'$ is a \emph{complete execution} if $\src(\exec)$ and $\trg(\exec)$ are quiescent states and $\exec'$ contains only active states.
A \emph{run} is a concatenation of complete executions executed in the same code context $\ocode$,
By abuse of notation, we use $\exec$ to denote runs as well as executions.
In case we want to make the code context $\ocode$ of the run explicit, we write $\ocode \inctx \exec$. 

The \emph{minimal} and \emph{maximal depths} of a non-empty execution $\exec$, denoted by 
$\minexecdepth(\exec) = \min \{ \execdepth(\sstate) \mid \sstate \in \execstates(\exec)\}$ 
and 
$\maxexecdepth(\exec) = \max \{ \execdepth(\sstate) \mid \sstate \in \execstates(\exec)\}$ 
are the minimal, respectively, maximal depths of any of the states it contains.

A well formed execution $\exec'$ is an \emph{invocation} in an execution $\exec$ if there exist transition $\tr$ and $\tr'$ such that 
$\exec=\_ \cdot \tr \cdot  \exec' \cdot  \tr' \cdot \_$, where $\statedepth(\src(\tr))=\statedepth(\trg(\tr'))$ and $\minexecdepth(\exec')=\statedepth(\src(\tr))+1$.  
We refer to $\statedepth(\trg(\tr))$ as the \emph{depth} of the invocation $\exec'$ and denote it by $\execdepth(\exec')$.
Note that according to this definition, the depth of an invocation that results from calling a contract method on a quiescent state is one.

\paragraph{Traces}
We define an \emph{event} as a pair $\event=\ST{\objectid, \action}$, 
	consisting of  
	an \emph{object} $\objectid$, 
	and a \emph{primitive command} $\action$.
Each transition $\transition$ can be transformed to an event by $\event(\transition)=\ST{\objectid(\transition), \cmd(\transition)}$.
The \emph{object} and \emph{primitive command} of an event are can be retrieved with $\objectid(\event)$ and $\action(\event)$, respectively.
A \emph{trace} is a sequence of events, denoted by $\trace$.
The trace matching an execution $\exec$ is received by point-wise application of $\event(\cdot)$ on all the transitions in $\exec$, denoted $\trace(\exec)$.
We denote by $\trace|_\objectid$ the maximal subsequence of $\trace$ comprised of events whose object is $\objectid$.


\section{Execution Equivalence}\label{Se:EqExec}

We define two notions of equivalence of \emph{runs} (sequences of \emph{complete} executions) with respect to a given (arbitrary) object $\objectid$: \emph{final-state  equivalence} and \emph{conflict equivalence}.

\subsection{Object-Final-State Equivalence}
\ignore{
\begin{definition}
Executions $\exec_1$ and $\exec_2$ are \emph{final-state equivalent} if they start in the same state, and finish in the same state:
$$\exec_1 \fseq \exec_2 \iff \src(\exec_1) = \src(\exec_2) \land \trg(\exec_1)=\trg(\exec_2)$$
\end{definition}
}

\begin{definition}\label{De:obj-fin-eq}
Runs $\ocode_1 \inctx \exec_1$ and $\ocode_2 \inctx \exec_2$ are \emph{object-final-state equivalent for an object $\objectid$},
denoted by $\ocode_1 \inctx \exec_1 \fseqo{\objectid} \ocode_2 \inctx \exec_2$,  
if their respective first and last states have the same store for  $\objectid$, and their code contexts agree on~$\objectid$: 
$$
\begin{array}{l}
\ocode_1 \inctx \exec_1 \fseqo{\objectid} \ocode_2 \inctx \exec_2   \iff  {}\\
\qquad \SelStore(\src(\exec_1))(\objectid) = \SelStore(\src(\exec_2))(\objectid) \land \SelStore(\trg(\exec_1))(\objectid)=\SelStore(\trg(\exec_2))(\objectid)   \land \ocode_1(\objectid) = \ocode_2(\objectid)
\end{array}
$$
\end{definition}

\begin{remark}
Note that in \Cref{De:obj-fin-eq} it may be that $\ocode_1 \neq \ocode_2$ as long as $\ocode_1$ and $\ocode_2$ map $\objectid$ to the same code.
\end{remark}

\subsection{Object Conflict-Equivalence}


\paragraph{Conflicts}
Two primitive commands $\action$ and $\action'$ \emph{conflict},
denoted by  $\Conflict(\action,\action')$, 
if both access the same field and at least one of these accesses  is a write.

\begin{remark}
Recall that we assume that object states are comprised of integer-typed data members (fields); thus it is possible to detect the heap locations they access from the (syntactic) primitive commands that they execute.
Furthermore, recall that different objects contain different fields and an object cannot directly access the fields of another; thus if two primitive commands conflict then they must be executed by the same active object.
\end{remark}



\paragraph{Modular Well-Formed Executions}

We extend the definition of well-formed executions to \emph{modular} well-formed executions, which allow to 
replace the subexecution resulting from a method invocation $\SCallC{x}{\objectid'}{e}$ by  an assignment of arbitrary  value to $x$.
We refer to such an  (implicit) transition as a \emph{havoc transition}. 
Intuitively, a havoc transition allows to safely overapproximate the only effect that an object $\objectid$ may observe from the invocation of a method on an object $\objectid'$.
\begin{definition}
A  \emph{havoc transition}, denoted by $\sstate\rightsquigarrow\sstate'$, is a pair of states $\B{\sstate,\sstate'} \in  \State \times \State  $  such that
$\sstate=\defostate{\objectid}{\SCallC{x}{\objectid'}{e}}{\reg}{\store}$ and
$\sstate'=\defostate{\objectid}{\cdone}{\reg[\res \mapsto n]}{\store}$  for any values of $\objectid$, $\objectid'$, $e$, $\reg$, $\Stk$, $\store$ and $n$.
\end{definition}

\begin{definition}
A \emph{modular well-formed execution} is   a finite sequence of transitions coming from $\Tr$ such that 
for any consecutive transitions $\transition_1\transition_2$ it contains, 
the target state of $\transition_1$ and the source state of $\transition_2$ are either
(i) equal, i.e., $\trg(\transition_1)=\src(\transition_2)$, 
or (ii) induce a havoc transition, $\trg(\transition_1) \rightsquigarrow \src(\transition_2)$.
By abuse of notations, we use $\exec$  to denote modular well-formed executions too.
A sequence of transitions $\exec$ is a \emph{complete modular well-formed run} if it is modular well-formed and its first and last states are quiescent.
Such a run  is a \emph{complete modular well-formed execution} if, in addition, all other states are active.

\end{definition}

\paragraph{Projected Executions} 
\changed{We define what is the projection of an execution on an object, called \emph{projected execution}. 
The definition readily generalizes   to any sequence of transitions.}

\begin{definition}
Let $\exec$ be an execution.
The \emph{projected execution of $\exec$ on $\objectid$}, denoted by $\exec|_\objectid$,  
is the subsequence of $\exec$ comprised of the transitions whose active object is $\objectid$.
\end{definition}

\begin{lemma}
For any execution $\exec$ and object $\objectid$ it holds that  $\trace(\exec|_\objectid)=\trace(\exec)|_\objectid$.
\end{lemma}

\changed{
	It will be useful in the following sections to consider modular well-formed executions yielded by a projection on an object $\objectid$. 
	The following proposition states that it is possible to reverse that projection, namely, to find a minimal well-formed execution that, when projected, yields the modular well-formed execution we started with.
}
\begin{proposition}\label{Prop:ModularToWellFormed}
Let $\ocode \inctx \exec$ be a modular well-formed execution  where all   transitions have the same active object $\objectid$.
Then there is a context $\ocode'$ such that $\ocode'(\objectid)=\ocode(\objectid)$ 
 and an execution $\ocode' \inctx \exec'$ 
 such that $\trace(\exec)$ is a subsequence  in $\trace(\exec')$ 
 and $\ocode' \inctx \exec'$ is well-formed.
\changed{In addition, such an execution $\ocode' \inctx \exec'$ is minimal in the sense that $\trace(\exec|_{\objectid})=\trace(\exec'|_{\objectid})$}
\end{proposition}

\changed{
	Finally, we present the definition for execution conflict equivalence with respect to an object $\objectid$. 
}
 \begin{definition}
Let $\ocode_1 \inctx \exec_1$ and $\ocode_2 \inctx \exec_2$ be modular well-formed runs and 
$\trace_1= \trace(\exec_1|_\objectid)$ and $\trace_2= \trace(\exec_2|_\objectid)$ be their the traces of their respective projections on $\objectid$. 
$\ocode_1 \inctx \exec_1$ and $\ocode_2 \inctx \exec_2$ are \emph{object conflict-equivalent for an object $\objectid$}, denoted $\exec_1 \ceq^{\objectid} \exec_2$, if: 
\begin{enumerate}[(i)]
\item 
$\ocode_1(\objectid) = \ocode_2(\objectid)$, 
\item 
$|\trace_1| = |\trace_2|$, 
\item 
there exists a permutation 
$\cep:\{ 1..|\trace_1|\}\rightarrow\{ 1..|\trace_1|\}$  
such that:  
\begin{enumerate}
\item for any $i\in\{1..|\trace_1|\}$ it holds that $\trace_1(i)=\trace_2(\cep(i))$, 
and 
\item for any $i,j\in\{1..|\trace_1|\}$, if   $\Conflict(\action(\trace_1(i)), \action(\trace_1(j)))$
then $i<j \iff \cep(i)<\cep(j)$.
\end{enumerate}
\end{enumerate}
\end{definition}

\section{Correctness conditions}\label{Sec:CorrectnessConditions}

In this section we give a formal definition for two notions of the ECF property: \emph{final-state ECF} and \emph{conflict ECF}.
We start by formally defining callbacks and callback-freedom in executions.

\newcommand{\cb}[2]{\mathbb{CB}_{#2}(#1)}
\newcommand{\ncb}[2]{{\mathbb{CB}^{f}_{#2}(#1)}}



A stack frame $\stk$ \emph{is a callback frame (to object $\SelObject(\stk)$)} in a stack $\Stk$ 
if there exist stack frames $\stk'$ and $\stk''$ such that
$\Stk =\_ \stk'' \_\,\stk' \_\, \stk \_$ and $\SelObject(\stk)=\SelObject(\stk'')$, but $\SelObject(\stk)\neq\SelObject(\stk')$.
A stack $\Stk$ contains a \emph{callback (to $\objectid$)}, denoted by $\cb{\Stk}{\objectid}$, if it contains a callback frame.
A state $\sstate$  contains a \emph{callback (to $\objectid$)}, denoted by $\cb{\sstate}{\objectid}$,  if its stack does, and 
an execution $\exec$ contains a \emph{callback (to $\objectid$)}, denoted by $\cb{\exec}{\objectid}$,  if it contains a state $\sstate$ such that $\cb{\sstate}{\objectid}$.
A stack resp. state resp. execution is \emph{callback-free (for $\objectid$)}, denoted by $\ncb{\Stk}{\objectid}$ resp. $\ncb{\sstate}{\objectid}$ resp. $\ncb{\exec}{\objectid}$, 
if it does not contain a callback.

We now define what it means for an execution to be \emph{effectively} callback free, or \emph{ECF}, with respect to a given object $\objectid$.
Note that ECF is a property of both an execution $\exec$ and some object $\objectid$.
Specfically, we are  interested in the case where $\exec$ has a callback to $\objectid$.
\begin{definition}
A well-formed complete execution $\ocode \inctx \exec$ is \emph{equivalently effectively callback-free for an object $\objectid$}, denoted by $\DECFo[\subFS]{\objectid}$, if there is a well-formed callback-free run $\ocode' \inctx \exec'$, which is final-state equivalent for $\objectid$ to $\exec$:
$$
\ocode \inctx \exec \vDash \DECFo[\subFS]{\objectid} \iff \exists \ocode',\exec'. \, 
\ocode \inctx\exec \fseqo{\objectid} \ocode' \inctx\exec'
\land
\ncb{\exec'}{\objectid}  
$$
\end{definition}

We say that the execution $\ocode' \inctx \exec'$ is a \emph{witness} for $\exec$ being a $\DECFo[\subFS]{\objectid}$ execution.


Checking $\DECFo[\subFS]{\objectid}$ is difficult in practice, and undecidable in general for models with an infinite state.
We describe a stronger definition of ECF, based on conflict-equivalence, called $\CECF$, 
which permits an efficient algorithm for checking it.
Interestingly, even though executions in our model do not allow for concurrency, callbacks can be thought of as allowing to express a limited subset of concurrent executions. 
In fact, the ECF property in our model is analogous to serializability in models that permit concurrency.
Using this analogy, \emph{invocations} are analogous to \emph{transactions}.
We show what this means to reorder invocations in the sequential semantics of $\PLname$.

In general terms, $\CECF$ requires to find a callback-free execution which is conflict-equivalent to the execution with the callbacks.
Conflict-equivalence requires that the trace of the callback-free execution is a permutation of the trace of the original execution.
It is thus useful to start with a characterization of the legal permutations of an execution.
Firstly, the permutation may not break program order of contract code.
That is, the permutation must retain the ordering of events whose transitions are part of the same invocation $\exec$,
 and their state have the same depth as $\execdepth(\exec)$.
Secondly, we want to allow permutations that remove callback invocations from their original call location, and sets them to execute in a quiescent state, and still receive a modular well-formed execution.


\usetikzlibrary{decorations.pathmorphing}
\tikzset{snake/.style={decorate, decoration=snake},
q/.style={circle,fill=black,draw,minimum size=0.1cm,inner sep=0pt,label={$\epsilon$}},
                    o1/.style={circle,fill=black,draw,minimum size=0.1cm,inner sep=0pt},
                    o2/.style={circle,red,fill=red,draw,minimum size=0.1cm,inner sep=0pt},
                    o3/.style={circle,fill=black,draw,minimum size=0.1cm,inner sep=0pt},
                    }

\begin{figure}

\begin{tikzpicture}[->,
                    ]

\node[q]    (q1) {};
\node[o1]   (1o1) [right = 0.8cm of q1] {};
\node[o1]   (2o1) [right = 0.8cm of 1o1] {};
\node[o2]   (1o2) [right = 0.8cm of 2o1] {};
\node[o2]   (2o2) [right = 0.8cm of 1o2] {};
\node[o1]   (3o1) [right = 0.8cm of 2o2] {};
\node[q]    (q2) [right = 0.8cm of 3o1] {};
\node[o3]   (1o3) [right = 0.8cm of q2] {};
\node[q]    (q3) [right = 0.8cm of 1o3] {};

\path[draw, thick, every node/.style={anchor=south}]
(q1) edge node {\smaller{$\trid^1_1$}} (1o1);

\path[draw, dashed, every node/.style={anchor=south}]
(1o1) edge node {$\cdots$} (2o1);

\path[draw=red, thick, every node/.style={anchor=south}]
(2o1) edge node {\red{\smaller{$\trid^1_2$}}} (1o2);

\path[draw=red, dashed, every node/.style={anchor=south}]
(1o2) edge node {\red{$\cdots$}} (2o2);

\path[draw=red, thick, every node/.style={anchor=south}]
(2o2) edge node {\red{\smaller{$\trid^m_2$}}} (3o1);

\path[draw, thick, every node/.style={anchor=south}]
(3o1) edge node {\smaller{$\trid^n_1$}} (q2)
(q2) edge node {\smaller{$\trid^1_3$}} (1o3);

\path[draw, dashed, every node/.style={anchor=south}]
(1o3) edge node {$\cdots$} (q3);

\end{tikzpicture}

\hspace{+2.6cm}
\begin{tikzpicture}[->,
                    ]

\node[q]    (q1) {};
\node[o1]   (1o1) [right = 0.8cm of q1] {};
\node[o1]   (2o1) [right = 0.8cm of 1o1] {};
\node[o1]   (3o1) [right = 2.6cm of 2o1] {};
\node[q]    (q2) [right = 0.8cm of 3o1] {};
\node[o3]   (1o3) [right = 0.8cm of q2] {};
\node[q]    (q3) [right = 0.8cm of 1o3] {};
\node[o2]   (1o2) [right = 0.8cm of q3] {};
\node[o2]   (2o2) [right = 0.8cm of 1o2] {};
\node[q]    (q4) [right = 0.8cm of 2o2] {};

\path[draw, thick, every node/.style={anchor=south}]
(q1) edge node {\smaller{$\trid^1_1$}} (1o1);

\path[draw, dashed, every node/.style={anchor=south}]
(1o1) edge node {$\cdots$} (2o1);

\path[draw, snake, every node/.style={anchor=south}]
(2o1) -- (3o1);

\path[draw, thick, every node/.style={anchor=south}]
(3o1) edge node {\smaller{$\trid^n_1$}} (q2)
(q2) edge node {\smaller{$\trid^1_3$}} (1o3);

\path[draw, dashed, every node/.style={anchor=south}]
(1o3) edge node {$\cdots$} (q3);

\path[draw=red, thick, every node/.style={anchor=south}]
(q3) edge node {\red{\smaller{$\trid^1_2$}}} (1o2);

\path[draw=red, dashed, every node/.style={anchor=south}]
(1o2) edge node {\red{$\cdots$}} (2o2);

\path[draw=red, thick, every node/.style={anchor=south}]
(2o2) edge node {\red{\smaller{$\trid^m_2$}}} (q4);

\end{tikzpicture}

\hspace{+2.6cm}
\begin{tikzpicture}[->,
                    ]

\node[q]    (q1) {};
\node[o1]   (1o1) [right = 0.8cm of q1] {};
\node[o1]   (2o1) [right = 0.8cm of 1o1] {};
\node[o1]   (3o1) [right = 2.6cm of 2o1] {};
\node[q]    (q2) [right = 0.8cm of 3o1] {};
\node[o2]   (1o2) [right = 0.8cm of q2] {};
\node[o2]   (2o2) [right = 0.8cm of 1o2] {};
\node[q]    (q3) [right = 0.8cm of 2o2] {};
\node[o3]   (1o3) [right = 0.8cm of q3] {};
\node[q]    (q4) [right = 0.8cm of 1o3] {};

\path[draw, thick, every node/.style={anchor=south}]
(q1) edge node {\smaller{$\trid^1_1$}} (1o1);

\path[draw, dashed, every node/.style={anchor=south}]
(1o1) edge node {$\cdots$} (2o1);

\path[draw, snake, every node/.style={anchor=south}]
(2o1) -- (3o1);

\path[draw, thick, every node/.style={anchor=south}]
(3o1) edge node {\smaller{$\trid^n_1$}} (q2);

\path[draw=red, thick, every node/.style={anchor=south}]
(q2) edge node {\red{\smaller{$\trid^1_2$}}} (1o2);

\path[draw=red, dashed, every node/.style={anchor=south}]
(1o2) edge node {\red{$\cdots$}} (2o2);

\path[draw=red, thick, every node/.style={anchor=south}]
(2o2) edge node {\red{\smaller{$\trid^m_2$}}} (q3);

\path[draw, thick, every node/.style={anchor=south}]
(q3) edge node {\smaller{$\trid^1_3$}} (1o3);

\path[draw, dashed, every node/.style={anchor=south}]
(1o3) edge node {$\cdots$} (q4);
\end{tikzpicture}

\hspace{-2.9cm}
\begin{tikzpicture}[->,
                    ]

\node[q]    (q1) {};
\node[o2]   (1o2) [right = 0.8cm of q1] {};
\node[o2]   (2o2) [right = 0.8cm of 1o2] {};
\node[q]    (q2) [right = 0.8cm of 2o2] {};
\node[o1]   (1o1) [right = 0.8cm of q2] {};
\node[o1]   (2o1) [right = 0.8cm of 1o1] {};
\node[o1]   (3o1) [right = 2.6cm of 2o1] {};
\node[q]    (q3) [right = 0.8cm of 3o1] {};
\node[o3]   (1o3) [right = 0.8cm of q3] {};
\node[q]    (q4) [right = 0.8cm of 1o3] {};

\path[draw=red, thick, every node/.style={anchor=south}]
(q1) edge node {\red{\smaller{$\trid^1_2$}}} (1o2);

\path[draw=red, dashed, every node/.style={anchor=south}]
(1o2) edge node {\red{$\cdots$}} (2o2);

\path[draw=red, thick, every node/.style={anchor=south}]
(2o2) edge node {\red{\smaller{$\trid^m_2$}}} (q2);

\path[draw, thick, every node/.style={anchor=south}]
(q2) edge node {\smaller{$\trid^1_1$}} (1o1);

\path[draw, dashed, every node/.style={anchor=south}]
(1o1) edge node {} (2o1);

\path[draw, snake, every node/.style={anchor=south}]
(2o1) -- (3o1);

\path[draw, thick, every node/.style={anchor=south}]
(3o1) edge node {\smaller{$\trid^n_1$}} (q3)
(q3) edge node {\smaller{$\trid^1_3$}} (1o3);

\path[draw, dashed, every node/.style={anchor=south}]
(1o3) edge node {$\cdots$} (q4);

\end{tikzpicture}
\caption{The callback reorder process. 
The first graph represents the original execution, which contains a callback (in red).
The three other graphs represent all possible callback free executions.
Red marks the moved callback transitions.
Wave edges indicate that a call was replaced with a havoc transition.
}\label{fig:CallbackExtraction}
\end{figure}
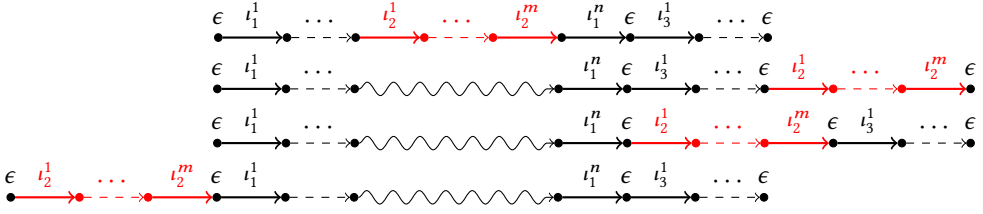

When we permute a trace such that a callback invocation is removed from its original place, we replace the call transition leading to the callback with a havoc transition.
An example can be seen in \Cref{fig:CallbackExtraction}, showing all legal permutations of a trace that has a callback-free execution with havoc transitions.

\changed{
\begin{remark}
A havoc transition is not part of the permuted trace; it is merely used to justify the execution which is no longer well-formed as we defined earlier.
\end{remark}
}

Using modular well-formed executions, we can formally define the $\CECF$ property for executions, $\DECF[\subC]$:
\begin{definition}
A well-formed complete execution $\ocode \inctx  \exec$ is \emph{conflict-equivalently effectively callback-free for an object $\objectid$},
denoted by $\ocode \inctx \exec \vDash \DECFo[\subC]{\objectid}$, 
if there is a modular well-formed callback-free run $\ocode  \inctx  \exec'$ which is object-conflict-equivalent to $\exec$ for $\objectid$:
$$
\ocode \inctx  \exec \vDash \DECFo[\subC]{\objectid} 
\iff 
\exists \exec'. \,
\ocode \inctx  \exec \ceq^{\objectid} \ocode \inctx  \exec'
\land
\ncb{\exec'}{\objectid} 
\,.
$$
\end{definition}

It is easy to prove that conflict equivalence implies final-state equivalence (see, e.g., \cite{BOOK:BHG87}).
Thus, it can be concluded that $\CECF$ implies $\FSECF$ \changed{as, by using \Cref{Prop:ModularToWellFormed},}{} we can use the same witness for $\exec$ being $\DECF[\subC]$ to prove that $\exec$ is also $\DECF[\subFS]$.
\begin{theorem}\label{Thm:CECFImpliesFSECF}
Let $\exec$ be a well-formed complete execution. 
If  $\exec \vDash \DECFo[\subC]{\objectid}$ then  $\exec \vDash \DECFo[\subFS]{\objectid}$.
\end{theorem}
\changed{
\begin{proof}
Since $\ocode \inctx \exec \vDash \DECFo[\subC]{\objectid}$, there is a callback-free modular well-formed complete execution $\ocode \inctx \exec'$ 
 such that $\ocode \inctx \exec \ceq^{\objectid} \ocode \inctx \exec'$.
It is easy to see that since $\ncb{\exec'}{\objectid}$ then $\exec'|_{\objectid}$ is too a callback-free and modular well-formed run. 
From an immediate generalization of \Cref{Prop:ModularToWellFormed} to runs, 
we conclude that for $\ocode \inctx \exec'|_{\objectid}$ there is 
 a context $\ocode'$ 
 and a well-formed run $\ocode' \inctx \hat{\exec}$ 
 such that $\ocode'(\objectid)=\ocode(\objectid)$, 
  and $\trace(\exec'|_{\objectid})=\trace(\hat{\exec}|_{\objectid})$. 
  It should be noted that the equality of the trace projected on $\objectid$ implies callback-freedom for $\objectid$ is retained: 
  namely, $\ncb{\hat{\exec}}{\objectid}$.
Conflict equivalence implies final-state equivalence (\cite{BOOK:BHG87}),
 hence, $\src(\exec)=\src(\exec')$ and $\trg(\exec)=\trg(\exec')$.
Also, because conflict ordering in primitive commands of $\objectid$ is retained between $\exec'|_{\objectid}$ and $\hat{\exec}|_{\objectid}$, we conclude that $\exec' \fseq^{\objectid} \hat{\exec}$, therefore $\src(\exec')=\src(\hat{\exec})$ and $\trg(\exec')=\trg(\hat{\exec})$. 
It can be concluded then that since (i) $\src(\exec)=\src(\hat{\exec})$ and $\trg(\exec)=\trg(\hat{\exec})$, 
 (ii)  $\ocode(\objectid)=\ocode'(\objectid)$, and (iii) $\hat{\exec}$ is well-formed, then:
  $\exec \fseq^{\objectid} \hat{\exec}$.
  making $\ocode' \inctx \hat{\exec}$ a witness proving $\exec \vDash \DECFo[\subFS]{\objectid}$.
\end{proof}
}

Finally, as we are also interested in ECF as a property of objects (\SECF{}), 
we extend the definitions of $\FSECF$ and $\CECF$ to  objects (\SECF{\subFS} and \SECF{\subC}) instead of executions (\DECF{\subFS} and \DECF{\subC}), 
which we refer to as \emph{static ECF}.
\ignore{\emph{projected executions}, which include only transitions that pertain to a single object.
\begin{definition}
Let $\exec$ be an execution and $\objectid$ be an object which is the active object of a state $\sstate\in\execstates(\exec)$.
The \emph{projected execution of $\objectid$} is an execution $\exec_o$ whose trace contains only events $\event\in\trace(\exec)$ such that $\objectid(\event)=\objectid$.
\end{definition}
\changed{Projected executions are used to define the two notions of static ECF (\SECF{}):}
}

\begin{definition}
An object $\objectid$ is $\SECF[\subFS]$   
if for every complete execution $\ocode \inctx \exec$ it holds that \mbox{$\ocode \inctx \exec \vDash \DECFo[\subFS]{\objectid}$}.
$\objectid$ is $\SECF[\subC]$   
if for every complete execution $\ocode \inctx \exec$ it holds that $\ocode \inctx \exec \vDash \DECFo[\subC]{\objectid}$ .



\end{definition}

\section{Decidability}\label{Sec:Decidability}
This section discusses the decidability of verifying \ECF. 
Using Rice Theorem (see, e.g.,~\citet{ullman}), 
it is easy to show that 
verifying \SECF, namely, statically verifying whether  all executions of an object are  
\FSECF{}  or   $\CECF$, is an undecidable problem.
Interestingly, 
checking $\FSECF$ for a single execution ($\DECFo[\subFS]{\objectid}$) is also undecidable. 

\changed{
\begin{theorem}
Given an execution $\ocode \inctx \exec$, checking if it is $\DECFo[\subFS]{\objectid}$ is undecidable.
\end{theorem}
\begin{proof}
We show a reduction from the halting problem. 
$\PLname$ is Turing-complete, thus we encode the operation of a Turing machine $M$ as a command $c\in\Cmd$. 
In \Cref{fig:decfo-decidability-reduction}, we present the code of a contract \texttt{A}.
The store of \texttt{A} has a single field $X$ initialized as $0$, and a single argument denoted $arg$. 
The field $X$ is unchanged by $M$. 
 We also write the code of a contract \texttt{B}.
The method of \texttt{A} is separated to 3 branches. If $X\neq 0,1$ then the method returns without any effect. 
If $X=1$, then $X$ is updated to $2$ and the method returns.
If $X=0$, then $X$ is updated to $1$ and if the argument is equal $0$, we execute the TM $M$ and update $X$ to 2 if and when $M$ finished running.
If the argument is not equal $0$, we call contract \texttt{B}. 
If right after \texttt{B}'s execution $X$ is not equal $2$, then$X$ is updated to $3$.
The code of \texttt{B} is calling to \texttt{A}. Therefore, when \texttt{A} calls \texttt{B}, \texttt{B} always creates a callback to \texttt{A}.
\begin{figure}  
\centering
  \begin{subfigure}{2.5in}
    \centering
    \begin{footnotesize}
    \begin{alltt}
\begin{tabbing}
XX\=XX\=XX\=XX\=XX\=XX\=XX\=XX\=XX\=XX\=XX\=\kill
Object A \+\\
  int X \\
  M's fields \\
  enter (arg) \\
    if X == 0 then \+\\
      X := 1 \\
      if arg != 0 then \+\\
      	B() \\
      	if X != 2 then X := 3 \- \\
      else \+\\
        run M() \\
        X := 2 \-\-\\
    else if X == 1 then \+\\
      X := 2 \-\\
  return \-\\
\end{tabbing}
    \end{alltt}
    \end{footnotesize}
      \end{subfigure}
  \quad
  \begin{subfigure}{2.5in}
    \centering
      \begin{footnotesize}
    \begin{alltt}
    \begin{tabbing}
    XX\=XX\=XX\=XX\=XX\=XX\=XX\=XX\=XX\=XX\=XX\=\kill
Object B \+\\
  enter \+\\
    A(0) \-\\
  return \\
\end{tabbing} 
     \end{alltt}
    \end{footnotesize}
    \usetikzlibrary{arrows}
    \begin{tikzpicture}[line cap=round,line join=round,>=latex,x=0.5cm,y=0.5cm]
    \begin{scriptsize}
    \draw [->,line width=0.5pt] (5.2,2.1) -- (6.7,2.1); 
    \draw [->,line width=0.5pt] (6.7,2.1) -- (8.5,2.1); 
    \draw [->,line width=0.5pt] (8.5,2.1) -- (9.6,3.02); 
    \draw [->,line width=0.5pt] (8.5,2.1) -- (9.6,1.4); 
    \draw [->,line width=0.5pt] (9.6,1.4) -- (10.2,0.85); 
    \draw [->,line width=0.5pt] (10.2,0.85) -- (11.5,0.85); 
    \draw [->,line width=0.5pt] (11.5,0.85) -- (12.8,0.85); 
    \draw [->,line width=0.5pt] (12.8,0.85) -- (13.5,1.5); 
    \draw [->,line width=0.5pt] (13.5,1.5) -- (13.5,0.75); 
    \draw [->,line width=0.5pt] (13.5,1.5) -- (14.6,1.5); 
    \draw [->,line width=0.5pt] (14.6,1.5) -- (16.3,1.5); 
    \draw [->,line width=0.5pt] (16.3,1.5) -- (17,1.5); 
    \draw [->,line width=0.5pt] (9.6,3.) -- (10.6,3.); 
    \draw [->,line width=0.5pt] (10.6,3.) -- (12.,3.); 
    \draw [->,line width=0.5pt] (12.,3.) -- (12.8,3.); 
    \draw (3.6,2.2) node[anchor=north west] {$q_{X=0}$};
    \draw (6.0,2.0) node[anchor=north west] {$e_{A_{X=0}}$};
    \draw (6.6,2.8) node[anchor=north west] {$X\!:=\!1$};
    \draw (5.1,2.8) node[anchor=north west] {$A(a)$};
    \draw (8.,2.6) node[anchor=north west, rotate=43] {$a\!=\!0$};
    \draw (8.5,2.2) node[anchor=north west,  rotate=-41  ] {$a\!\neq\! 0$};
    \draw (9.3,3.8) node[anchor=north west] {$M()$};
    \draw (10.5,3.8) node[anchor=north west] {$X\!:=\!2$};
    \draw (12.5,2.9) node[anchor=north west] {$r_A$};
    \draw (9.6,2.) node[anchor=north west] {$B()$};
    \draw (9.8,0.85) node[anchor=north west] {$e_B$};
    \draw (10.2,1.65) node[anchor=north west] {$A(0)$};
    \draw (10.85,0.9) node[anchor=north west] {$e^{cb}_{A_{X=1}}$};
    \draw (11.3,1.6) node[anchor=north west] {$X\!:=\!2$};
    \draw (12.3,0.85) node[anchor=north west] {$r_B$};
    \draw (13.5,0.85) node[anchor=north west] {$r_A$};
    \draw (13.3,2.3) node[anchor=north west] {$X\!\neq\! 2$};
    \draw [->,line width=0.5pt] (12.8,3) to[out=50,in=60] (5.2,3.2);
    \draw (14.5,2.3) node[anchor=north west] {$X:=3$};
    \draw (16.2,1.3) node[anchor=north west] {$q_{X=3}$};
    \draw [->,line width=0.5pt] (13.5,0.75) to[out=-85,in=-110] (5.2,3.2);
    \draw (3.6,3.3) node[anchor=north west] {$q_{X=2}$};

    \draw [fill=black] (5.2,2.1) circle (1.0pt); %
    \draw [fill=black] (6.7,2.1) circle (1.0pt); 
    \draw [fill=black] (8.5,2.1 ) circle (1.0pt); %
    \draw [fill=black] (9.62,3.02) circle (1.0pt);
    \draw [fill=black] (9.62,1.4) circle (1.0pt);
    \draw [fill=black] (10.5,3.) circle (1.0pt);
    \draw [fill=black] (10.2,0.85) circle (1.0pt);
    \draw [fill=black] (11.5,0.85) circle (1.0pt);
    \draw [fill=black] (12.8,0.85) circle (1.0pt);
    \draw [fill=black] (13.5,1.5) circle (1.0pt);
    \draw [fill=black] (13.5,0.75) circle (1.0pt);
    \draw [fill=black] (12.,3.) circle (1.0pt);
    \draw [fill=black] (12.8,3.) circle (1.0pt);
    \draw [fill=black] (14.6,1.5) circle (1.0pt);
    \draw [fill=black] (16.3,1.5) circle (1.0pt);
    \draw [fill=black] (17,1.5) circle (1.0pt);
    \draw [fill=black] (5.2,3.2) circle (1.0pt); 
    \end{scriptsize}
    \end{tikzpicture}
\end{subfigure}

      \caption{The codes of two objects \texttt{A} and \texttt{B} showing that $\DECFo[\subFS]{\objectid}$ is undecidable, and a diagram showing the possible flows of the system. 
      Nodes marked $q$ are quiescent states, with the current value of \texttt{A}'s $X$ variable in subscript.
      Nodes starting with $e_o$ indicate entry to an object $o$, and $r_o$ a return from an object $o$.
      The notation $e^{cb}_o$ indicates the call is a callback. 
      Entry nodes are marked with the value of $X$ for better readability.
      Syntactic references to the code (assignments, calls, conditions) appear on the edges.
      Missing call edges from quiescent states indicate that the resulting quiescent state is the same, hence, $q_{X=2}$ and $q_{X=3}$ are sinks.
      }\label{fig:decfo-decidability-reduction}
      \end{figure}
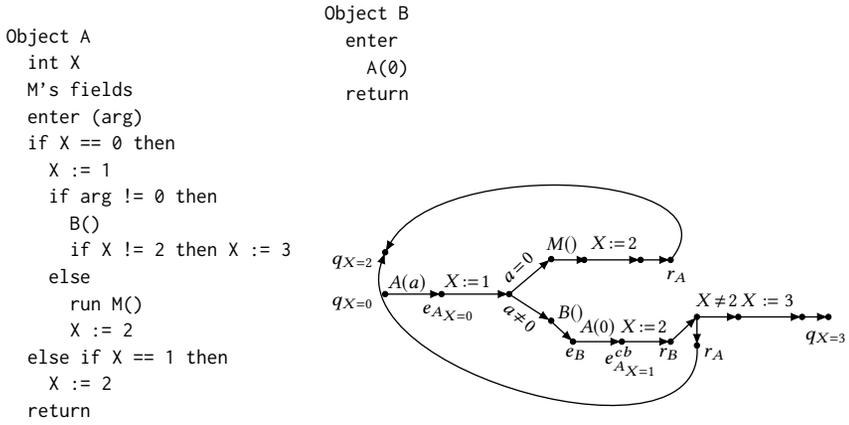
We consider the execution $\exec$ starting from the initial state in which \texttt{A}'s $X$ field is equal $0$, with $arg\neq 0$.
This execution calls the object \texttt{B}, which calls back to \texttt{A}, and in the callback the value of $X$ is set to $2$.
We show that $\exec$ is $\DECFo[\subFS]{\objectid}$ if and only if $M()$ halts.
For the `if' direction, we note that if $M()$ halts, then the execution of $A(0)$ from the state where $X=0$ leads to $X$ being set to $2$ right after $M$'s run finished, and that execution has no callbacks, as required. 
For the `only if' direction, we note that the only callback-free execution that starts from $X=0$ and ends with $X=2$ is the call \texttt{A}(0), which is the execution that executes $M()$, and it is a legal execution only if $M()$ halts. 
The reason that this must be the only execution, is that for any choice of input argument $arg\neq 0$ and context $\ocode$ that maps a different code for \texttt{B}, does not allow reaching the required final state $X=2$ unless callbacks are used. 
If $\ocode(\texttt{B})$ still calls back to \texttt{A} then clearly the resulting execution is not an eligible ECF witness.
If $\ocode(\texttt{B})$ does not call back to \texttt{A}, then $X$ is updated to $3$. 
Therefore, when $X=3$, any subsequent call to \texttt{A} cannot modify $X$, and in particular $X=2$ is not reachable.
\end{proof}

}
\changed{
\begin{figure}
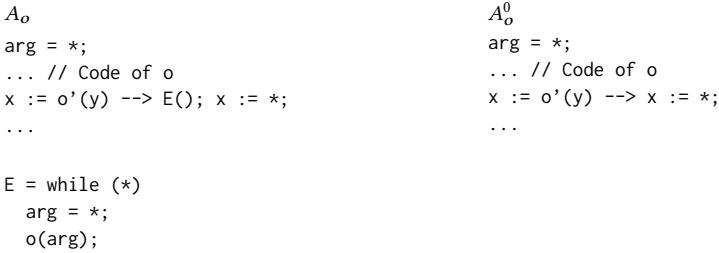

\begin{subfigure}[t]{2.5in}
\begin{footnotesize}
$A_\objectid$
\begin{alltt}
\begin{tabbing}
XX\=XX\=XX\=XX\=XX\=XX\=XX\=XX\=XX\=XX\=XX\=\kill
  arg = *; \\
  ... // Code of o \\
  x := o'(y) --> E(); x := *; \\
  ... \\
\end{tabbing}
\end{alltt}
\begin{alltt}
\begin{tabbing}
XX\=XX\=XX\=XX\=XX\=XX\=XX\=XX\=XX\=XX\=XX\=\kill
E = while (*) \+\\
    arg = *; \\
    o(arg); \-\\
\end{tabbing} 
\end{alltt}
\end{footnotesize}
\end{subfigure}
\begin{subfigure}[t]{2.5in}    
\begin{footnotesize}
$A^0_\objectid$
\begin{alltt}
\begin{tabbing}
XX\=XX\=XX\=XX\=XX\=XX\=XX\=XX\=XX\=XX\=XX\=\kill
  arg = *; \\
  ... // Code of o \\
  x := o'(y) --> x := *; \\
  ... 
\end{tabbing}
\end{alltt}
\end{footnotesize}
\end{subfigure}
\caption{The construction of automatons $A_\objectid$ and $A^0_\objectid$. 
The notation \texttt{-->} indicates that a command in the code of $\objectid$ is replaced with another.}\label{Fig:PDAConstruction}
\end{figure}
}
In contrast, checking $\CECF$ for a single execution (\DECF[\subC]) is obviously decidable, as we can enumerate all of the permutations of a particular input trace.

Thus, we focus on verifying \SECF, namely, statically verifying whether  all executions of an object are 
\FSECF{}  or  $\CECF$, 
where the domains of the object variables are restricted to finite sets.
Hence, such objects can be modeled with a pushdown-automaton ($\PDA$). 
\changed{
  Such a PDA for an object $\objectid$ is able to simulate any modular well-formed execution $\ocode \in \exec$ where the active object of all states in $\exec$ is $\objectid$. 
  We denote this construction $A_\objectid$. Its code is shown in the left side of \Cref{Fig:PDAConstruction}. 
  It executes the code of $\objectid$ with a non-deterministically chosen argument, and replaces every call to an external object with a sequence of arbitrary calls to $\objectid$ (the code in \texttt{E}), and a non-deterministic choice of the return value of the call.
}
We begin with a rather simple lemma that shows $\FSECF$ of objects is indeed decidable in this model. 
We assume that 
 variables may take values coming from a finite domain.


\begin{lemma}
\label{prop:FSECFDecidable}
Let $\objectid$ be an object, assuming a finite domain for variables. 
Then there is an algorithm that decides if $\objectid$ is $\SECF[\subFS]$.
\end{lemma}
\begin{proof}
\changed{	
We consider the pushdown automaton $A_\objectid$ and an automaton that allows only callback-free behaviors for $\objectid$, $A^0_\objectid$. 
Both are shown in \Cref{Fig:PDAConstruction}. 
Each execution of $A_\objectid$ consists of first choosing non-deterministically an argument $arg$ with which the object $\objectid$ is called. 
The automaton simulates the steps of $\objectid$ with the only difference being in invocations of other objects, e.g. $\objectid'$.
This command is replaced with the code of the ``environment'' object $E$, that performs a sequence of calls to $\objectid$ with an arbitrary argument chosen in each iteration. 
The sequence length is arbitrary and may be also $0$, i.e., no callbacks.
After running the callbacks, the automaton non-deterministically chooses a return value, stored in the original variable intended to store the return value of $\objectid'$, which in the figure is the variable $x$.

The automaton that allows only callback-free behaviors for $\objectid$ is simpler: 
 the only difference is that it does not run the code of $E$. 
In fact, $A^0_\objectid$ is a finite state machine.

We now utilize the result by~\cite{bouajjani1997reachability} for reachability of a regular set of configurations of a pushdown automaton.
Here, the set of configurations is $\State$, and the subset of configurations which we are interested in is the one with with an empty stack, i.e. $\{\ST{}\} \times \Sigma_\objectid$ where $\Sigma_\objectid$ is the set of $\objectid$'s possible states, which is finite and thus regular. 

We consider an arbitrary pair of states $\sigma_1,\sigma_2\in\Sigma_\objectid$. 
We first check if there is an execution of $A_\objectid$ starting from $\sigma_1$ and ending in $\sigma_2$. 
This can be done by checking if $\sigma_1$ is reachable from the initial state and as well as if $\sigma_2$ is reachable from $\sigma_1$.
As the path that shows reachability in $A_\objectid$ may not be callback-free, we check for reachability also in the finite state machine $A^0_\objectid$, in the same manner. 
Reachability in finite-state machines is known to be decidable. 
If there is no execution in $A^0_\objectid$ from $\sigma_1$ to $\sigma_2$, then the object $\objectid$ is not $\SECF[\subFS]$.
After checking for all pairs in $\Sigma_\objectid \times \Sigma_\objectid$ 
 that for each pair of states $(\sigma_1,\sigma_2)$ there is an execution of $A_\objectid$ starting in $\sigma_1$ and ending in $\sigma_2$ if and only if there is such an execution in $A^0_\objectid$, we verified that $\objectid$ is $\SECF[\subFS]$. 
The set of all pairs of states is finite, thus it describes a decision procedure for verifying $\SECF[\subFS]$ for an object $\objectid$.

}

\end{proof}

Showing the decidability of $\SECF[\subC]$ of objects is not that easy, because it requires
reasoning on permutations of events, which is not a regular property, even in the case of finite-state machines. 

\changed{
Our strategy for proving the decidability of $\SECF[\subC]$ will be as follows. 
As before, we consider executions of $A_\objectid$, which identify with the set of possible projected executions on $\objectid$.
We will show that it is enough to check subexecutions of $A_\objectid$ which involve at most two elements of the stack. 
Namely, it is the same as checking the set of executions of $A_\objectid$ with a limit of two to the depth of the stack, 
 where the initial state of the execution  may include both the initial states $\Sigma_\objectid$ of an execution of $A_\objectid$, and in addition any state $\sstate$ where a callback may be called.
 Explicitly, $\sstate$ is a state whose primitive command is a call. 
We describe such states using the set $\mathcal{I}_0=\{\sstate \mid \exists x,\oid',e. \SelCmd(\topstk(\sstate))=\SCallC{x}{\oid'}{e} \}$. 
  It is regular as we only consider the top element of the stack.
Finding for the set of states $\mathcal{I}_0$ 
 its subset of reachable states in $A_\objectid$ 
 can therefore be done using~\cite{bouajjani1997reachability}. 
 
We denote by $A^2_\objectid$ the automaton that is received by limiting the stack depth of $A_\objectid$ to two, and having initial states $\mathcal{I}_0\cup\Sigma_\objectid$. 
For $A^2_\objectid$, deciding if all its executions are $\DECFo[\subC]{\objectid}$ is decidable. 
We do so by constructing a monitor that receives as input execution traces of $A^2_\objectid$ and outputs an error if the executions are not $\DECFo[\subC]{\objectid}$. 
Then we will show that $M$ is a finite state machine and thus checking reachability of its error states is decidable.
Such a construction is possible because even though conflict-equivalence demands finding a permutation of a trace, which is not a regular property, the actual choice of permutations that have to be checked is much more limited.
\begin{lemma}
Let $\Pi$ denote the set of all executions of $A^2_\objectid$. 
There is an automaton $M$ that for each $\exec\in\Pi$, ends in an accepting state if $\exec \vDash \DECFo[\subC]{\objectid}$, and in a rejecting state otherwise.
\end{lemma}
\begin{figure}
\begin{footnotesize}
{\begin{tt}
\begin{tabbing}
XX\=XX\=XX\=XX\=XX\=XX\=XX\=XX\=XX\=XX\=XX\=\kill
$M(\exec)$ = \+\\
  d = 0, prefix = $(\emptyset,\emptyset)$, delayedCbs = $(\emptyset,\emptyset)$ \\
  for ($\event \in \trace(\exec)$) \+\\
  if ($\cmd(\event)=\SAssign{x}{F}$) \+ \\
    R := R$\cup \{F\}$ \- \\
  if ($\cmd(\event)=\SAssign{F}{x}$) \+ \\
    W := W$\cup \{F\}$ \- \\
  if ($\cmd(\event)=\SEnterV$ \&\& d=0) // Execution starts \+ \\
    R, W := $(\emptyset, \emptyset)$  \\
    d := d+1 \- \\  
  if ($\cmd(\event)=\SReturnV$ \&\& d=1) // Execution ends \+ \\
    assert((R,W) commutes with delayedCbs) \\
    R, W := $(\emptyset, \emptyset)$ \\
    d := d-1 \\
    return \- \\
  if ($\cmd(\event)=\SEnterV$ \&\& d=1) // Callback starts \+ \\
    assert((R,W) commutes with delayedCbs) \\
    prefix := (R(prefix) $\cup$ R, W(prefix) $\cup$ W) \\
    R, W := $(\emptyset, \emptyset)$ \\
    d := d+1 \-\\
  if ($\cmd(\event)=\SReturnV$ \&\& d=2) // Callback ends \+ \\
    if (!((R,W) commutes with prefix) || !((R,W) commutes with delayedCbs)) \+ \\
      delayedCbs := (R(delayedCbs) $\cup$ R, W(delayedCbs) $\cup$ W) \- \\
    R, W := $(\emptyset, \emptyset)$ \\
    d := d-1 
\end{tabbing}
\end{tt}
}
\end{footnotesize}
\caption{The code of $M$ which accepts an execution $\exec$ of $A^2_\objectid$ and verifies if it is $\DECFo[\subC]{\objectid}$}\label{fig:DECFForLevel2DecisionProcedure}
\end{figure}
\begin{proof}
We write the code of the automaton $M$ in \Cref{fig:DECFForLevel2DecisionProcedure}.
The automaton loops on each event in the trace of the execution. 
The automaton states consist of a depth variable {\tt d}, and two pairs of read and write sets: {\tt prefix} and {\tt delayedCbs}.
The sets {\tt R} and {\tt W} are updated in each command that reads from or writes to the object store.
 They are reset when a callback starts or ends.
The {\tt prefix} pair retains  the accumulated reads and writes in the invocations at depth $1$, 
 i.e. the first invocation of $\objectid$ and which we refer to sometimes as the \emph{main} invocation.
The {\tt delayedCbs} pair retains the accumulated reads and writes by callbacks that we choose to execute after the main invocation ends.

Intuitively, the monitor checks each time a pair of {\tt (R,W)} is finalized and before it is reset, if it satisfies conditions that will allow to find a conflict-equivalent execution.
For a callback execution, we check if the pair commutes with the prefix, i.e. all portions of the main invocation already executed.
 If it does not commute with the prefix, we mark it as a delayed callback, and update the {\tt delayedCbs} read and write sets.
  As we do not know whether delayed callbacks actually commute with the rest of the invocation, 
   and with future callbacks that may be executed before the main invocation, 
   we retain the conflict information in {\tt delayedCbs} to be checked later against any finalized {\tt (R,W)} sets that is belonging either to the main invocation, or to a callback that is chosen to be executed before the main invocation.
  Therefore, even in the case when a callback commutes with the prefix, but in which it does not commute with the delayed callbacks (which are `jumping over' the callback under consideration in order of execution),
   we have to try to execute it as a delayed callback as well. 
   Otherwise, the execution is surely not $\DECFo[\subC]{\objectid}$, and we move to an error state.
For a portion of the main invocation (the first or last one, or between callbacks), we check if it does not conflict with any of the relevant delayed callbacks. 
The {\tt delayedCbs} pair is always updated correctly with the currently delayed callbacks' read and write locations. 

It can be seen from the construction that if $M$ accepts, then we can build from $M$'s execution a witness for $\exec$ being $\DECFo[\subC]{\objectid}$ by ordering all non-delayed callbacks before the main invocation and the delayed callbacks after the main invocation, both of those in the order in which they appear in the original execution.
E.g., if callbacks $i_1$ and $i_2$ are both delayed callbacks, then if $i_1$ is executed before $i_2$ in $\exec$ then the ordering between them in the witness is not changed. The same argument applies for non-delayed callbacks.

In addition, if $\objectid$ is $\DECFo[\subC]{\objectid}$ then $M$ must accept.
 This is because any witness for $\exec$ being $\DECFo[\subC]{\objectid}$ is conflict-equivalent to the witness implicitly produced by $M$.
The reason for that is that in the witness, we can also identify a set of callbacks executed before the main invocation and a set of callbacks executed after it.
The internal ordering of callbacks that are executed before the main invocation does not matter as long as it does not reorder conflicts, and thus may in fact be the same as in the original execution, ditto for the other set of callbacks.
\end{proof}

We note that $M$ has a finite state: {\tt d} in $\{0,1,2\}$, and {\tt prefix} and {\tt delayedCbs} in $2^F \times 2^F$ where $F$ is the set of $\objectid$'s fields.
Therefore, $M$ is a finite state machine, and thus reachability is decidable. In particular, we can check if the error states are reachable. 
However, $M$ is a machine that works on any execution of any object. 
It is not difficult, though, to build $M$ as a monitor of $A^2_\objectid$. We denote this construction $M_{A^2_\objectid}$.
The result is still a finite state machine with decidable reachability, since $A^2_\objectid$, which has a bounded stack depth of 2, is also a finite state machine.

\begin{corollary}\label{corol:ConflictECFA2Decidable}
It is decidable to check if all executions of $A^2_\objectid$ are $\DECFo[\subC]{\objectid}$
\end{corollary}

We now show that the previous corollary implies that if all executions of $A^2_\objectid$ are $\DECFo[\subC]{\objectid}$, then then so are all executions of $A_\objectid$. 
In particular, it shows that $\objectid$ is $\SECF[\subC]$. 
Importantly, the converse is also true: if an object $\objectid$ is $\SECF[\subC]$, then all executions of $A^2_\objectid$ are $\DECFo[\subC]{\objectid}$. 
This is trivial to show: 
 for an execution $\exec$ of $A^2_\objectid$ such that $\exec \neg{\vDash} \DECFo[\subC]{\objectid}$ 
 it is easy to find a complete execution $\exec'$ in $A_\objectid$ of the form $\exec' = \exec_0 \hat{\exec} \exec'_0$ which is also not $\DECFo[\subC]{\objectid}$.
 Specifically, $\exec_0$  will be a subexecution that reaches the initial state of $\exec$, 
               $\hat{\exec}$ will be identical to $\exec$, only changing the stacks in $\exec$'s states to account for $\exec_0$'s transitions, 
               and $\exec'_0$ will complete the execution. 
  Any permutation of $\trace(\exec')$ will induce a conflict equivalence breaking permutation on $\hat{\exec}$, 
     whose conflicts are the same as those of $\exec$, 
     therefore $\exec'$ is not $\DECFo[\subC]{\objectid}$.

\begin{lemma}\label{lem:ConflictECFA2ToGeneral}
An object $\objectid$ is $\SECF[\subC]$ if all executions of $A^2_\objectid$ are $\DECFo[\subC]{\objectid}$.
\end{lemma}
\begin{proof}
We show that for every execution $\exec$ of $A_\objectid$ there is a witness $\exec'$ for it being $\DECFo[\subC]{\objectid}$.
The witness is built recursively. 
We initialize $\exec'=\exec$.
We denote $D=\maxexecdepth(\exec)$.
We start with the longest subexecution in $\exec_0 \subexec \exec$ which has $\maxexecdepth(\exec_0)=D$ and $\minexecdepth(\exec_0)=D-1$.
By ignoring the first $D-2$ stack frames, 
 $\exec_0$ can be seen as an execution of $A^2_\objectid$, 
 thus it is $\DECFo[\subC]{\objectid}$. 
 Hence, by taking the witness in $A^2_\objectid$ and returning the first $D-2$ stack frames, 
 we have a subexecution $\exec'_0 \ceq^\objectid \exec_0$ which is callback-free: $\minexecdepth(\exec'_0)=\maxexecdepth(\exec'_0)=D-1$. 
 We update $\exec'$ by replacing the transitions of $\exec_0$ with $\exec'_0$. 
The process continues by taking in each step the longest and deepest subexecution which has a callback and replacing it with a callback-free subexecution which is conflict-equivalent to it. 
In each step, the resulting $\exec'$ is conflict-equivalent to $\exec$.
The premise of the lemma ensures that this process will not stop until $\exec'$ is callback-free.
\end{proof}

From \Cref{corol:ConflictECFA2Decidable} and \Cref{lem:ConflictECFA2ToGeneral} we immediately get the following result:
\begin{theorem}
  Let $\objectid$ be an object, assuming a finite domain for variables. 
Then there is an algorithm that decides if $\objectid$ is $\SECF[\subC]$.
\end{theorem}
}

\section{Object-level analysis}\label{Sec:Modularity}
While the ECF property is capable of detecting unwanted executions which do not satisfy it,
it can be further used  for modular analysis of objects.
We will show that in environments in which objects are encapsulated, 
we can consider ECF for executions of a single object only, to help simplify object-level analysis.

We define the notion of a \emph{most general client (MGC)} in our model.
The most general client for an object $\objectid$, $MGC_\objectid$, is an external program that works on a system that includes a single object in the store. 
The store $\store$ of $MGC_\objectid$ contains a single object $NoCB(\objectid)$, which is built based on the original object $\objectid$.
Every invocation of an object of the form $\SCallC{x}{\objectid'}{e}$ in $\objectid$ is replaced with a non-deterministic choice of the value of $x$ as returned by the call, in correspondence with the definition of havoc transitions in \Cref{Sec:CorrectnessConditions}.
Furthermore, $MGC_\objectid$ is allowed to repeatedly call $NoCB(\objectid)$ with any parameter and in any order.
As such, the semantics of $MGC_\objectid$ soundly approximate all executions of the object $\objectid$ (see, e.g.,~\citet{GotsmanYangICALP11}), while every execution in $MGC_\objectid$ is in fact a projected, callback-free execution of $\objectid$.

We show that, if the object $\objectid$ is ECF in the general model, then any object-level assertion can be soundly verified on $MGC_\objectid$. 
This is because, all reachable states of an object $\objectid$ in any arbitrary code context $\ocode$ that does not change $\objectid$'s code, are reachable in the system containing only $NoCB(\objectid)$.
\begin{theorem}\label{Thm:ObjLevelAnalysisSoundness}
Let $R$ be the set of all states of the object $\objectid$ in a quiescent state,
and let $R_0$ be the set of all states of the object $\objectid$ in a run of $MGC_\objectid$.
If $\objectid$ is $\SECF[\subFS]$ then $R_0\supset R$.
%
\end{theorem}

Clearly, the theorem holds if the object is shown to be $\SECF[\subC]$, since this implies it is also $\SECF[\subFS]$ (see \Cref{Thm:CECFImpliesFSECF}).


This analysis is not overly imprecise, since in real environments, such as Ethereum, we could simulate such behaviors.
This is particularly correct since in Ethereum the store of the objects is updateable, and new objects may be added to the system.


Importantly, the analysis simply assumes ECF, and does not require to prove it: 
An alternative formulation of \Cref{Thm:ObjLevelAnalysisSoundness} is assuming that the runtime system enforces $\DECF$ on all executions.
In that case, the analysis is still sound. 
It is not unreasonable to assume such a dynamic analysis of $\DECF$, 
because we found out an efficient method to verify it, presented in \Cref{Sec:Dynamic}. 

We illustrate \Cref{Thm:ObjLevelAnalysisSoundness} using the example shown in \Cref{Fi:DaoContract}. 
We implemented \code{Dao} as a class in Dafny~\cite{Leino2010}  
with two methods: \code{deposit} and \code{withdrawAll}, 
whose pre- and post- conditions capture the object invariant which should be valid after every execution of the \code{Dao} object. 
Primarily, we wish to ensure that the data elements in the \code{credit} map are not negative, and that the sum of all these elements is equal to the balance of the \code{Dao}. 
We model the \code{pay} method without the recursive call to \code{withdrawAll}, but annotate it as possibly modifying (any field of) the \code{Dao} object. 
This annotation generalizes the possible behaviors of the \code{Dao} object without the ECF property. 
With such weak assumptions, and perhaps unsurprisingly, Dafny fails to verify the postconditions for both the original and the fixed versions of the \code{Dao} object. 
In contrast, when we assume that the ECF property holds, technically by adding a postcondition to \code{pay} which ensures that the previously read fields of the \code{Dao} object, i.e.,  \code{balance} and \code{credit[o]}, are not modified, Dafny is able to establish the post condition. 
\Cref{Thm:ObjLevelAnalysisSoundness} implies the fixed \code{Dao} contract respects the given specification when executed using the original runtime system and the original DAO  object respects the specification if it is executed on a runtime system which enforces ECF.

\section{Dynamic Verification}\label{Sec:Dynamic}
We describe a sound procedure for verifying the \DECF[\subC] property dynamically.
More precisely, for each execution, it checks for every object that participates in the execution, if the subsequence of the transitions that pertain only to that object (the \emph{projected execution}) is $\CECF$.
We assume the existence of an interpreter or virtual machine implementing the semantics defined in \Cref{Sec:Preliminaries}.
Below is a description of the data structures used by the algorithm, as well as the instrumentation of the object code to maintain these data structures. We then present a higher-level description of the algorithm, followed with pseudo-code and a complexity analysis.
We use the example presented in the overview section in \Cref{Fi:DaoContract} to explain the procedure.

The general structure of the procedure is that the instrumentation step starts every time we exit a quiescent state, and ends when we reach the next quiescent state. 
Once instrumentation has completed, 
the algorithm runs on the instrumented structures and returns whether all projected executions derived from the execution are ECF. 
The procedure repeats each time we enter an active state.

\subsection{Data Structures}
A \emph{segment} is a data structure that captures metadata about a portion of the execution's states.
This portion consists of a sequence of adjacent transitions, whose top stack frames have the same active object. 
That is, an invocation of a different object marks the beginning of a new segment, as well as returning from an invocation to a caller invocation which is executed in the context of a different object.
	However, a call from one object to itself does not break the current segment (This is motivated by the definition of callbacks in \Cref{Sec:CorrectnessConditions}).
In simpler terms, a new segment is defined each time the active object changes, 
either when we push a stack frame with a different object, 
or pop a stack frame such that the new top frame has a different active object.
We show how segments are determined in the instrumentation in \Cref{fig:Instrumentation}, using hooks on calls and returns.

\begin{example}\label{Exmp:Execution}
In the example DAO contract in \Cref{Se:Overview}, an attack execution consists of 6 segments:
(1) the first invocation of \code{withdrawAll}, lines 1-3;
(2) an invocation of \code{pay}, lines 3-5;
(3) the second invocation of \code{withdrawAll}, lines 1-3;
(4) a full invocation of \code{pay}, lines 3-4,7;
(5) the second invocation of \code{withdrawAll}, line 5;
(6) the first invocation of \code{withdrawAll}, line 5;
\end{example}

\begin{definition}[Segments]
A \emph{segment} $\segment$ is representative of a maximal sequence of adjacent transitions pertaining to the same object.
A segment $\segment=\ST{\readset{},\writeset{},\depth{},\indexInExec{}}$ contains information about fields accessed in the segment, 
denoted $\readset{\segment}$ and $\writeset{\segment}$ for the read- and write- sets, respectively.
In addition, a segment contains information about the depth of the invocation (denoted $\depth{\segment}$), which is equal to the depth of the transitions' states.
Last, the index in the execution (denoted $\indexInExec{\segment}$), is strictly increasing according to order of creation of the segments. 
\end{definition}

The primary metadata saved in each segment is the read and write sets of the fields of the object that were accessed by commands executed in the transitions that pertain to the segment.
 Other metadata includes the depth of the invocations in the stack, and an index to maintain the order of the segments in the execution. 

\begin{example}\label{Exmp:Segments}
We write down the segments that pertain to the \code{DAO} object in the overview example of the attack execution,
in the same order as they appear in the execution:
 $$\begin{array}{ll}
 	\segment_1=\ST{\SET{\icredit[\code{Attacker}],\ibalance},\SET{\ibalance},0,1} &
 	\segment_2=\ST{\SET{\icredit[\code{Attacker}],\ibalance},\SET{\ibalance},1,2} \\
 	\segment_3=\ST{\SET{},\SET{\icredit[\code{Attacker}]},1,3} &
 	\segment_4=\ST{\SET{},\SET{\icredit[\code{Attacker}]},0,4}
 	\end{array}
 	$$
\end{example}

An execution can be represented as a linear sequence of segments. 
Furthermore, from these segments 
we can determine 
the invocations that the execution contains.
\begin{remark}
Segments can be used as an alternative representation of executions and invocations, 
that generalize data saved by a sequence of transitions.
In this section only, we redefine the notions of executions and invocations to refer to segments instead of transitions.
\end{remark}
\begin{definition}[Executions, Invocations, and Callbacks]
An \emph{execution} can be represented using a sequence of its instrumented segments $\exec=\ST{\segment_1,\ldots,\segment_n}$.
We can access the j'th segment of the execution using $\exec(j)=\segment_j$. We trivially have that $\indexInExec{\segment_j}=j$.
An \emph{invocation} is a sequence of segments $\inv=\ST{\segment^{\inv}_1,\ldots,\segment^{\inv}_k}$ 
	such that there is a number $d$ for which all of the following holds:
	\begin{enumerate}
		\item $\forall \segment\in\inv. \depth{\segment}=d$ (all segments of the invocation are in the same depth).
		\item $d > \depthNoArg(\exec(\indexInExecNoArg(\segment^{\inv}_1)-1))$ (the first segment before the first segment in the invocation has lower depth, proving it is indeed the beginning of an invocation).
		\item $d > \depthNoArg(\exec(\indexInExecNoArg(\segment^{\inv}_k)+1))$ (the first segment after the last segment in the invocations has lower depth, proving it is indeed the end of an invocation).
		\item $\forall j. \indexInExec{\segment^{\inv}_1} < j < \indexInExec{\segment^{\inv}_k} \implies \depth{\exec(j)}\geq d$ (the invocation does not end before the last segment, that is all segments of depth $d$ in the given range belong to the same invocation).
	\end{enumerate}
As all segments included in the invocation has the same depth $d$, we denote the \emph{depth of an invocation} by $\depthNoArg(\inv)=d$.
We say that an invocation $\inv$ is a \emph{callback} in another invocation $\inv'$ (denoted $\inv\subexec\inv'$) if $\indexInExec{\inv(1)}>\indexInExec{\inv'(1)} \land \indexInExec{\inv(1)}<\indexInExec{\inv'(|\inv'|)}$. 
\end{definition}
\begin{remark}
Unlike the definition of invocations in \Cref{Sec:Preliminaries},
here invocations capture only the transitions in the same depth as the depth of its first transition, and not transitions in higher depth.
This allows to define $\depthNoArg(\inv)$ for an invocation $\inv$.
\end{remark}
\begin{example}\label{Exmp:Invocations}
In the attack execution presented in \Cref{Se:Overview},
the first invocation of \code{withdrawAll} is $\inv_{wd_1}=\ST{\segment_1,\segment_4}$, and the second invocation is $\inv_{wd_2}=\ST{\segment_2,\segment_3}$.
$\inv_{wd_2}$ is a callback of $\inv_{wd_1}$: $\inv_{wd_2}\subexec\inv_{wd_1}$.
\end{example}

We associate with each segment in depth $>1$ a \emph{prefix-set} and \emph{suffix-set} of all segments in the caller that precede, or respectively, proceed it:
\begin{definition}[Prefix and Suffix segments]
Let a set of segments representing an invocation $\inv=\ST{\segment^i_{\it{caller}}}$, 
and a single segment $\segment_{\it{cb}}$ with $\depthNoArg(\segment_{\it{cb}})>\depthNoArg(\inv)$ 
and $\indexInExecNoArg(\segment_{\it{cb}}) \in \SET{\indexInExec{\inv(1)},\ldots,\indexInExec{\inv(|\inv|)}}$. 
 We define for $\segment_{\it{cb}}$ its \emph{prefix and suffix sets relatively to a caller $\inv$} by partitioning the segment in $\inv$ to segments whose index in the execution is smaller than the index of the callback segment $\segment_{\it{cb}}$ (prefix), and segments whose index in the execution is larger than it (suffix):
$$
\begin{array}{lll}
\prefix{\inv,\segment_{\it{cb}}} & = & \{\segment_{\it{caller}}\in\inv \mid \indexInExec{\segment_{\it{caller}}} < \indexInExec{\segment_{\it{cb}}} \}
\\
\suffix{\inv,\segment_{\it{cb}}} & = & \{\segment_{\it{caller}}\in\inv \mid \indexInExec{\segment_{\it{caller}}} > \indexInExec{\segment_{\it{cb}}} \}
\end{array}
$$
\end{definition}
\begin{example}\label{Exmp:PrefixSuffixSegments}
The prefix and suffix segments of $\segment_2$ and $\segment_3$ with respect to $\inv_{wd_1}$ are:
$$
\begin{array}{rl}
\prefix{\inv_{wd_1},\segment_2} = \prefix{\inv_{wd_1},\segment_3} & =\SET{\segment_1} \\
\suffix{\inv_{wd_1},\segment_2} = \suffix{\inv_{wd_1},\segment_3} & =\SET{\segment_4}
\end{array}
$$
\end{example}

The instrumentation process creates the segments and the invocations. 
We show pseudo-code of the instrumentation procedure in \Cref{fig:Instrumentation}.
\begin{figure}	
\begin{lstlisting}[mathescape=true,basicstyle=\ttfamily\scriptsize]
Segment { Obj, Caller, R, W, D, I }
Invocation { Caller, Obj }

Init():
 execution := ()
 curSegment := $\bot$
 invocations := Map<Invocation -> $\overline{\texttt{Segment}}$>

UponInvocation(object):
 if fromQuiescent // Procedure starts
  Init()

 if object != curSegment.Obj
  caller := fromQuiescent ? TopInvocation : curSegment.Caller
  inv := Invocation(caller, object)
  AddSegment(object, inv, curSegment.D+1, curSegment.I+1)

UponReturn(object):
 caller := toQuiescent ? TopInvocation : curSegment.Caller.Caller
 if caller.Obj != object
  AddSegment(object, caller, curSegment.D-1, curSegment.I+1)

 if caller == TopInvocation // End of instrumentation step
  CheckECFForAllObjects() // Run the algorithm, and finish procedure

AddSegment(object, caller, D, I):
  segment := Segment(object, caller, {}, {}, D, I)
  Append(execution, segment)
  Append(invocations[caller], segment)
  curSegment := segment

UponObjectVarRead(object, F):
 curSegment.R[F] := 1

UponObjectVarWrite(object, F):
 curSegment.W[F] := 1
\end{lstlisting}
\caption{
Instrumentation procedures, implemented as hooks called upon call commands, return commands, and object variable read/write access command. Generates the \code{execution}, which is a list of segments, and \code{invocations}, a map of invocation identifiers to an invocation object keeping the caller of an invocation and the list of segments that are part of the invocation in the same depth.
The top-level invocation is identified as \code{TopInvocation}.
}\label{fig:Instrumentation}
\end{figure}


The basic check on segments is the commutativity check. We define segment commutativity using read and write sets. 
We will show that we actually check commutativity of a segment with either a prefix or suffix segment.
As the prefix/suffix segments are sets of segments, 
the read and write sets of prefix/suffix segments are a union of the respective read and write sets of all the segments contained in the prefix or suffix segment. 
\begin{definition}[Commutative Segments]\label{def:CommutativeSegments}
Segments $\segment_1$ and $\segment_2$ commute, denoted by $\Commute{\segment_1}{\segment_2}$, if:
$$
\Commute{\segment_1}{\segment_2} \eqdef \readset{\segment_1}\cap\writeset{\segment_2}=\emptyset \land \readset{\segment_2}\cap\writeset{\segment_1}=\emptyset \land \writeset{\segment_1}\cap\writeset{\segment_2}=\emptyset  
$$
If segments $\segment_1$ and $\segment_2$ do not commute, we denote this by $\nCommute{\segment_1}{\segment_2}$.
\end{definition}
\begin{example}\label{Exmp:Commutativity}
In the attack execution presented in \Cref{Se:Overview}, indeed we have that $\nCommute{\segment_2}{\segment_1}$,
as $\readset{\segment_1}\cap\writeset{\segment_2}=\SET{\ibalance}$.
Similarly, $\nCommute{\segment_3}{\segment_1}$ because of $\icredit[o]$, 
as $\readset{\segment_1}\cap\writeset{\segment_3}=\SET{\icredit[o]}$,
and therefore also $\nCommute{\segment_2}{\segment_4}$. 
However, $\segment_3$ does commute with $\segment_4$: $\Commute{\segment_3}{\segment_4}$.
\end{example}

\subsection{Algorithm}
We start with a high-level description of the algorithm.
The algorithm is called every time the system reaches a quiescent state, working on the last complete execution.
The algorithm generates a relation of invocations that defines 
constraints on the ordering of invocations in different stack depths, similar to a `happens-before'~\cite{Lamport:1978:TCO:359545.359563} relation.
We name this relation the \emph{invocation order constraint (IOC) graph}.
For example, if a segment $\segment$ of a callback invocation $\inv_{\it{cb}}$ is not commuting with its prefix with respect to one of its calling invocations $\segment_{\it{caller}}$ (i.e., $\nCommute{\segment}{\prefix{\inv_{\it{caller}},\segment}}$), then we add the constraint that the invocation of the caller has to occur before the callback: $\hb{\inv_{\it{caller}}}{\inv_{\it{cb}}}$. 
The IOC relation of invocations is thus defined as:
$$
\begin{array}{lll}
\hb{\inv}{\inv'} & \eqdef & (\inv'\subexec\inv \land \exists \segment\in\inv'. \nCommute{\segment}{\prefix{\inv,\segment}}) \\
& & \lor (\inv\subexec\inv' \land \exists \segment\in\inv. \nCommute{\segment}{\suffix{\inv',\segment}})
\end{array}
$$
\begin{example}\label{Exmp:IOC}
The IOC relation of the attack execution in \Cref{Se:Overview} can be easily calculated with the previous metadata given in examples \ref{Exmp:PrefixSuffixSegments} and \ref{Exmp:Commutativity}.
We have that $\hb{\inv_{wd_1}}{\inv_{wd_2}}$ as $\inv_{wd_2}\subexec\inv_{wd_1}$ and for $\segment_2\in\inv_{wd_2}$, $\nCommute{\segment_2}{\prefix{\inv_{wd_1},\segment_2}}$.
Similarly, $\hb{\inv_{wd_2}}{\inv_{wd_1}}$ as for $\segment_2\in\inv_{wd_2}$, $\nCommute{\segment_2}{\suffix{\inv_{wd_1},\segment_2}}$.
\end{example}

After the IOC relation is defined, the algorithm considers the graph induced by this relation, 
and checks it has no cycles. 
A cycle in the graph could appear if, for example, there is a callback invocation and some caller invocation that contains it, for which there is both (1) a segment that does not commute with its prefix with respect to the caller; and (2) a segment that does not commute with its suffix with respect to the caller.
As each vertex in this graph represents an invocation, the topological sorting returns an ordering of the invocations, which is $\CECF$. 
We are merely interested if there is such a topological sorting, that is, if the IOC relation does not contain a cycle.
\begin{theorem}
Let $\exec$ be an execution and let $\it{Inv}$ be a map of the instrumented invocations to their segments.
We denote by $\hb{}{}$ the IOC on $\it{Inv}$.
If $\hb{}{}$ has no cycle, then $\exec$ is $\CECF$.
\end{theorem}
\begin{proof}
We assume $\hb{}{}$ has no cycle. 
We take a total order $\hb{}{^t}$ of $\it{Inv}$ induced by the transitive closure on $\hb{}{}$.
From $\hb{}{^t}$ we build a run $\exec'$ such that every invocation in $\it{Inv}$ starts in a quiescent state in the order determined by $\hb{}{^t}$.
$\exec'$ is conflict-equivalent to $\exec$.
To show this, we consider two transitions $\transition_1$ and $\transition_2$ which conflict in $\exec$.
If $\transition_1$ and $\transition_2$ are both captured in the same segment during instrumentation, then their ordering is kept in $\exec'$ which only reorders invocations. In particular, the program order of invocations is kept.
The same argument applies when $\transition_1$ and $\transition_2$ are not captured by the same segment, but their respective segments are both part of the same invocation.
In the general case, $\transition_1$ and $\transition_2$ each belong to different segments, pertaining to different invocations.
In that case, their ordering in $\exec'$ is kept as $\hb{}{^t}$ respects that conflict.
%
\end{proof}

\begin{example}
In continuation to our running example, it is immediate that the IOC relation of the attack execution on the DAO object has a cycle: $\hb{\inv_{wd_1}}{\hb{\inv_{wd_2}}{\inv_{wd_1}}}$. 
Therefore, the algorithm cannot determine the attack execution is $\CECF$.
But indeed, the attack execution is not ECF, thus it cannot be $\CECF$.
\end{example}

We already saw in \Cref{Sec:Modularity} that due to state encapsulation, ECF is a modular property.
Therefore, the procedure may either check ECF for the entire execution, by searching for a cycle in the full IOC relation, or to check ECF for one object at a time.
	A modular ECF check can be done by projecting the relation only on invocations of the object under examination.
To align with the actual implementation of the algorithm, 
	we chose to present it in its modular version here as well.

We give the complete pseudo-code of the algorithm in \Cref{fig:DynamicMonitorAlgorithm}.
It begins with an additional step of preprocessing which is calculating the commutativity matrix of all segments against all prefix and suffix segments of all their enclosing invocations (invocations that directly or indirectly call the invocation in which the segment is included). 
The commutativity matrix assists in calculating the IOC relation.
We then iterate over all objects encountered in the execution, project the IOC relation on a single object in each iteration, and check if it has a cycle. 
If the check returns that it is a DAG, then we verified the projected execution is ECF. Otherwise, the cycle describes the invocations which cannot be moved, and helps identify the callbacks that cause the violation of ECF.

\begin{figure}
\begin{lstlisting}[basicstyle=\ttfamily\scriptsize]]
CheckECFForAllObjects():
 commute_matrix := CalculateCommutativityMatrix()
 hbRelation := CalculateIOCRelation(commute_matrix)
 for each unique object in execution:
  if not CheckECF(object):
   Print "Object " object " is not ECF"

CheckECF(object):
 // It is guaranteed that IOC applies only to invocations of the same object
 hbRelationO := project hbRelation on invocations of object only
 return isDAG(hbRelation)
		
CalculateCommutativityMatrix():
 matrix := new Map<Invocation, Segment -> Bool, Bool>
 for each inv in invocations, segment in execution
  if encloses(inv, segment)
   prefix := (s for s in inv where s.I < segment.I)
   suffix := (s for s in inv where s.I > segment.I)	
   prefixRS, prefixWS := Union(s.R for s in prefix), Union(s.W for s in prefix)
   suffixRS, suffixWS := Union(s.R for s in suffix), Union(s.W for s in suffix)
   prefixCommute := isCommutative(prefixRS, prefixWS, segment.R, segment.W)
   suffixCommute := isCommutative(suffixRS, suffixWS, segment.R, segment.W)

   if prefixCommute == False && suffixCommute == False
    Abort("Not ECF")
   matrix[inv, segment] := prefixCommute, suffixCommute

 return matrix

encloses(inv, segment):
  return inv.Obj = segment.Obj
	&& segment.I between first segment and last segment in inv

CalculateIOCRelation(commute_matrix):
  rel := new Map<Invocation, Invocation -> Bool>
  for each inv1, inv2 in invocations
	if encloses(inv1, first segment in inv2)
	  for each segment in inv2
		if commute_matrix[inv1, segment] == False, True
		  rel[inv1, inv2] := True

	if encloses(inv2, first segment in inv1)
	  for each segment in inv1
		if commute_matrix[inv2, segment] == True, False
		  rel[inv1, inv2] := True
\end{lstlisting}
\caption{Algorithm for verifying ECF of an execution.
The code of \code{isDAG} and \code{isCommutative} is not given. The definition of \code{isCommutative} is given according to Definition \ref{def:CommutativeSegments}.
}\label{fig:DynamicMonitorAlgorithm}
\end{figure}

\subsection{Complexity}
\subsubsection{Time.}
The instrumentation step adds a constant factor of work to the runtime.
To analyze the algorithm, we begin by looking at the preprocessing steps first.
Let $n$ denote the number of invocations and $m$ the number of segments ($n<m$).
	In addition, let $k$ denote the maximal number of object variables accessed in an object participating in the execution ($k<m$).
The \code{CalculateCommutativityMatrix} procedure loops on all invocations and all segments.
For each pair of an invocation and a segment, the \code{encloses} predicate can be implemented to take constant time.
The calculation of the prefix set and suffix set is taking time linear in the number of segments in an invocation, bounded by $m$.
The time to calculate the read and write sets of the prefix and the suffix set is linear in $k$.
Commutativity check, which involves checking set intersection, where our sets are implemented as associative arrays, is linear in $k$.
Thus the time of \code{CalculateCommutativityMatrix} is $O(nm(m+k))=O(nm^2)$.
For \code{CalculateIOCRelation}, we have a loop over pairs of invocations, and another pair of non-nested loops over segments in an invocation, giving $O(mn^2)$.
Projecting the IOC relation is linear in its size which is $O(n^2)$.
The \code{isDAG} check is linear in the size of the projected relation, which is bounded by $O(n^2)$. (The graph it represents has $O(n)$ vertices and $O(n^2)$ edges, and checking for a graph to be a DAG is $O(|V|+|E|)$).
In total, we have $O(nm^2)$.

\subsubsection{Space.}
The instrumentation adds $O(m)$ space for keeping the segments, and $O(nm)$ for keeping the invocations.
The commutativity matrix takes $O(nm)$ space, 
	and the IOC relation takes $O(n^2)$ space.
Therefore, the space complexity of the algorithm is $O(nm)$.


\section{Evaluation}\label{Sec:Evaluation}
We developed a prototype implementation for a dynamic monitor verifying ECF for \emph{Ethereum}.\footnote{The source code is available at \url{https://github.com/shellygr/ECFChecker}.}
For each execution, it checks if any of the participating contracts has a non-ECF (projected) execution, and outputs all detections of non-ECF executions.\footnote{The monitor was implemented on top of the Go~\cite{GoProgrammingLanguage} client for Ethereum, called \emph{geth}~\cite{Eth:Geth}, version 1.5.9.}
We ran our experiments by importing the entire blockchain from its inception on July 30, 2015 until March 30, 2017. \footnote{Without delving into the specifics of the blockchain paradigm, executions are organized in a structure called \emph{blocks}. Our primary experiment was to import the first 3,444,354 blocks of the main Ethereum blockchain.}
The host we used is a 64-bit Ubuntu 16.04 with two 2.2~GHz Intel Xeon E5-2699 processors (22 cores each with 2 threads per core) and 256~GB of RAM.
Both the instrumentation and the algorithm were integrated directly into the \emph{Ethereum Virtual Machine (EVM)} module using hooks, as described in \Cref{Sec:Dynamic}.


Our monitor operated in ``'detect-mode'' to avoid affecting the results, and for statistics gathering only.
However, it is trivial to change it to ``'prevent-mode'', that actively invalidates and reverts complete executions which are not ECF.
Had all the Ethereum clients used such a monitor by design, the DAO incident would have been avoided, along with the controversial hard fork. 
As our experiments prove, the false-positive rate of the monitor is minuscule: only 10 executions out of about 100 million were legitimately non-ECF. There is also no concern of performance impact, as the measured overhead of running the monitor was less than $3.5\%$. Furthermore, the benchmarks of the monitor were performed in an ideal environment that actually makes the overhead larger than it is in a normal environment. The reason behind it is that normal environments have additional overheads such as networking and disk accesses, which we disabled in order to scale our experiments.

\paragraph{Experiments.}
In \Cref{fig:ECFGeneralExperiment}, we show a short list of experiments conducted. 
We also included the number of contracts created as an additional metric of the blockchain.
The primary experiment was checking for ECF in all executions since the creation of blockchain until March 30, 2017. 
	Of note is that less than $0.01\%$ of the  executions were non-ECF.
	In the second experiment, we processed all executions starting from March 30, 2017 until June 23, 2017
	\footnote{The second experiment processed all blocks from block no. 3,444,355 to block no. 3,918,380.}.

It is interesting to compare the results of the first experiment, conducted on a snapshot of the Ethereum blockchain taken in Mar. 30th, 2017, which we used for benchmarks, to the second experiment, in which we let our modified client to process all newer executions, until June 23rd, 2017.
The number of non-ECF contracts decreased in both absolute quantity and in percentage of total executions. 
 The newer executions expose the maturity of the network, expressed in both the number of total contracts created (almost $150\%$ more contracts created in less than 3 months than in the entire existence of the blockchain, from August 2015 till the snapshot date), and the number of executions.\footnote{Assuming a new block is generated at an almost constant rate, there were 32 executions per block on average in the second experiment, compared with 23 executions in the first.}
 Moreover, the number of executions with callbacks increased significantly, indicating more complex contracts.
 	In the first experiment there were $128,670$ executions containing callbacks, and in the second experiment there were $155,668$ executions with callbacks. 
 	This amounts to a $641\%$ increase in the number of callbacks in the later period compared with the earlier period.
 	While the percentage of executions with callbacks is still only $1\%$ of all executions, 
 	the absolute number of executions with callbacks is large enough to indicate that callbacks are inevitable,
 	either because they are useful, or necessary.
	This means that contracts show an increasing use of callbacks, and thus more complex code, that may be prone to bugs resulting from unintended interaction between contracts.
In both experiments, the overall percentage of non-ECF executions out of the executions with callbacks, was $1.17\%$, and less than $0.01\%$ out of all executions.
\begin{figure}
$$
\begin{array}{lrrrrr}
\textbf{Blockchain} 	& \textbf{Date} & \textbf{Contracts}	& \textbf{Executions} & \textbf{Callbacks} & \textbf{Non-ECF (\%)} 	 \\
\text{Ethereum} 	& \text{30.VII.2015-30.III.2017} & $138,457$ & $81,097,421$ & $128,670$	& $3,315$\, ($0.004\%$)	 \\
\text{Ethereum} 	& \text{30.III.2017-23.VI.2017} & $203,859$	& $15,311,650$ 	& $155,662$	& $6$\, ($<0.001\%$)	 \\
\text{Eth. Classic}	& \text{30.VII.2015-29.VI.2017}	& $91,191$	& $32,494,464$ 	& $81,731$	& $2,288$ \, ($0.007\%$) \\
\end{array}
$$
\caption{
Experimental results.
We use dates to mark the portion of the blockchain checked in the experiment.
The Contracts column shows how many contracts were created (but not necessarily executed) in the relevant time period.
The Executions column records the number of method invocations and the 
Callbacks column shows how many of these invocations were callbacks.
Non-ECF column counts how many non-ECF executions were detected, and their percentage out of the total number of executions. 
}
\label{fig:ECFGeneralExperiment}
\end{figure}

\paragraph{Discussion of non-ECF examples.}
We present a list of all contracts that demonstrated non-ECF executions in \Cref{fig:ECFContracts}.
Contracts C2, C4  are related to the DAO. C2 is the original DAO~\cite{Eth:DAOAttack}. C4 is known as `The Dark DAO'~\cite{Eth:DarkDAO}, an object containing a copy of the DAO's code, as created by the attack (The mechanism of the DAO was such that, every withdrawal of funds, manifested in the form of a new object whose code is a copy, or `split', of the DAO code).
Contract C1 is an unrelated contract which suffered a vulnerability very similar to the DAO's. The vulnerability, also stemming from non-ECF behavior, was discovered during a security audit and disclosed shortly before the attack on the DAO~\cite{Eth:MakerDAO,Eth:NonECFMention}.
Contract C5 is an exercise published on the blockchain to demonstrate the DAO attack~\cite{Eth:DAOExercise}, and indeed a non-ECF execution was detected.
In some contracts, it is difficult to pinpoint the exact cause of the existence of non-ECF executions, as the only available code is EVM bytecode, which is not trivial to analyze and reverse.
We tried to connect these incognito contracts with their creators or users.
With this approach, we found evidence that C3 is also related to the DAO~\cite{Eth:0x34a}.

The contract at C7~\cite{Eth:0xbf7Creator} was traced back to Validity Labs~\cite{Validity}.
We contacted the authors and they provided us with the Solidity source-code of the contract~\cite{Eth:Validity}.
It was deliberately designed to have a DAO-style callback exploit, and was used in their training workshops to demonstrate its dangers.

\begin{figure}

\begin{minipage}{0.5in}
\begin{footnotesize}
\begin{alltt}
\begin{tabbing}
Object C \\

  Object Sender \\

  Method call(data, sender) \\
    if (Sender != nil) throw \\
    Sender = sender; ret = this.do(data); Sender = nil \\
\\
  Method do(data) \\
    ... // read Sender
\end{tabbing}
\end{alltt}
\end{footnotesize}
\end{minipage}
\caption{Pattern used by contracts C6, C8 and C9. 
\code{Sender} is initialized to \code{nil}.
\code{call} is a method that throws when \code{Sender} is not \code{nil}, and otherwise sets it, calls method \code{do}, and nullifies \code{Sender} afterwards.
}\label{fig:Ambisafe}
\end{figure}

\begin{figure}
$$
\begin{array}{llrrrrr}
\textbf{Name} & \textbf{Contract address} & \textbf{Execs.} & \textbf{Execs. w. cbs.} & \textbf{Non-ECF} & \textbf{Stack depth} \\
\multicolumn{7}{c}{\textit{Ethereum Network (ETH)}} \\
\text{C1} & $\small{0xd654bdd32fc99471455e...}$ &  $924$ & $143$ & $10$ & $3$ \\
\text{C2} & $\small{0xbb9bc244d798123fde78...}$ & $274,820$ & $103,064$ & $3,296$ & $2-146$ \\ 
\text{C3} & $\small{0x34a5451ef61a567ee088...}$ & $91$ & $8$ & $1$ & $46$ \\
\text{C4} & $\small{0x304a554a310c7e546dfe...}$ & $13,223$ & $2,812$ & 1 & $3$ \\
\text{C5} & $\small{0x59752433dbe28f5aa59b...}$ & $15$ & $6$ & $1$ & $3$ \\
\text{C6} & $\small{0x97361ea911d6348cf2af...}$ & $44$ & $42$ & $6$ & $2$ \\
\text{C7} & $\small{0xbf78025535c98f4c605f...}$ & $25$ & $22$ & $3$ & $3-9$ \\
\text{C8} & $\small{0x232f3a7723137ced12bc...}$ & $144$ & $142$ & $1$ & $2$ \\
\text{C9} & $\small{0x7c525c4e3b273a3afc4b...}$ & $35$ & $33$ & $2$ & $2$ \\

\multicolumn{7}{c}{\textit{Ethereum Classic Network (ETC)}} \\
\text{C1} &  $\small{0xd654bdd32fc99471455e...}$			& $850$ & $143$ & $10$		& $3$	\\
\text{C2} &  $\small{0xbb9bc244d798123fde78...}$		 & $195,428$ & $86,573$ & $805$			& $2-146$\\
\text{C3} &  $\small{0x34a5451ef61a567ee088...}$		 & $18$ & $9$ & $1$			& $46$	\\
\text{C4} &  $\small{0x304a554a310c7e546dfe...}$			& $14,150$ & $3,064$ & $1$		& $3$\\
\text{C10} & $\small{0xf4c64518ea10f995918a...}$		& $428$ & $177$ & $11$		& $42-122$\\
\text{C11} & $\small{0xb136707642a4ea12fb4b...}$	 & $2,582$ & $305$ & $201$			& $17-20$\\
\text{C12} & $\small{0x0e0da70933f4c7849fc0...}$	 & $5,330$ & $3,992$ & $1,259$		& $12-57$
\end{array}
$$
\vspace*{-3mm}
\caption{A sample of interesting Non-ECF contracts in Ethereum.
Contracts are given a name $\text{C1},\ldots,\text{C12}$, and are ordered chronologically, by the date of the first  non-ECF execution.
The Executions and Executions with callbacks columns show statistics on usage style. 
The Non-ECF column shows how many executions were detected as non-ECF. 
Stack depth column indicates the range of the depths of the non-ECF subexecutions.}
\label{fig:ECFContracts}
\end{figure}

\begin{figure}
$$
\begin{array}{llll}
					& \textbf{Time} & \textbf{Memory (max)} & \textbf{Memory (end)} \\
\text{Monitor off}	 & $16h 17m$ & $5.5GB$ & $803MB$ \\
\text{Monitor on}  & $16h 50m$ \, ($3.38\%$ \text{ overhead}) & $5.5GB$ \, ($0\%$) & $940MB$ \, ($17\%$ \text{ overhead}) \\
\end{array}
$$
\vspace*{-3mm}
\caption{Performance statistics. Benchmark experiment was importing the Ethereum main network blockchain, from its creation in July 30, 2015 until March 30, 2017. Compares the import with monitor on or off.}
\label{fig:Performance}
\end{figure}



The same high-level Solidity code of contracts C6, C8, and C9, were provided to us by their creators at Ambisafe~\cite{Eth:Ambisafe}.
The pattern used by these contracts gives rise to behaviors that are purposefully non-ECF.
We show a snippet illustrating the pattern in \Cref{fig:Ambisafe}.
This pattern is inherently non-ECF.\footnote{In the formal definition, it actually is ECF, because a call of a contract to itself is not a callback. 
The contracts under examination were discovered due to a deviation of our monitor's implementation from the full definition of ECF. However, this example can be fitted into a slightly modified pattern which is not-ECF even according to the full definition, by adding an intermediary contract between \code{call} and \code{do}.}
The method \code{do} assumes the value of \code{Sender} is not \code{nil}, but this only occurs in the context of an invocation of \code{call}. 
The purpose behind this behavior, is to have \code{Sender} act as a lock, protecting against unexpected callbacks.
Such a design may be avoided in presence of a monitor that allows only ECF executions.

The bottom part of the table in \Cref{fig:ECFContracts} shows non-ECF contracts found in~\citet{Eth:ETHClassic}.
Ethereum Classic (or ETC) is the continuation of the original Ethereum blockchain following the controversy of the hard-fork due to the DAO bug.
Until July 20, 2016, both blockchains, Ethereum and Ethereum Classic, contain the same executions, and thus the same non-ECF executions.
Our result and investigation show that all non-ECF executions discovered in the Ethereum Classic network are of copies of the DAO~\cite{Eth:DAOinETC}.

Generally, it is important to stress that: (1) there may be other non-ECF contracts, as crafting and deploying contracts that exploit non-ECF entails investment of real money, thus requires a strong incentive to do so; (2) attacking is harder as Ethereum employs (not bullet-proof) heuristics to limit callbacks; (3) a better playground may be the Ethereum TestNet on which we did not run the experiment, but may provide insight as a future work. 

The actual overhead measured by enabling the ECF monitor is given in \Cref{fig:Performance}.
We used the first experiment, where the blocks were imported, as a benchmark. 
Normally, there is an additional overhead of network download times, which can vary significantly.
The measured overhead is about $3.5\%$, when calculating the difference in time of importing the blockchain with the monitor off, and importing it with the monitor on.
We believe the actual overhead is even smaller in most realistic scenarios.
First, most clients import the blockchain using the network, which may cause unexpected latencies, unrelated to the monitor.
Additionally, the process was pointed to a directory created on a 200~GB RAM disk to improve the scalability of the experiment.\footnote{The Ethereum blockchain suffered a DoS attack~\cite{Eth:DoSAttack,Eth:DoSAttack2} affecting the blockchain in the range of block numbers 2.2M-2.7M, causing all peers participating in the blockchain to make frequent accesses to disk. Running on a RAM disk was necessary to minimize the runtime of experiments.}
Most clients use a physical disk and not a RAM disk. 
Even if the physical disk is an SSD drive, the experiment slows down significantly, and takes about 20h (18\% more than with a RAM disk). 

The additional memory footprint measured in the end of the import is about $140MB$, or $17\%$. 
It should be noted, that as the implementation is written in Go, which includes automatic garbage collection, the memory consumption varies between tests. The relative difference with the monitor on or off was consistent across repeated tests.
 The maximal memory used by the process is $5.5GB$ and is not related to the monitor.
High memory consumption occurred during the processing of one of the DoS attacks on the blockchain.

\section{Related Work}\label{Sec:RelatedWork}

\subsection{Modular Reasoning}
Modular reasoning is a topic which has been studied extensively
with the seminal works of~\citet{DBLP:journals/acta/Hoare72} and~\citet{DBLP:books/ph/Dijkstra76}.
For more recent studies on modularity we refer the readers to~\citet{DBLP:conf/fm/LeinoM05,DBLP:journals/jacm/BanerjeeN05}.

  \emph{Averroes}~\cite{RW:AliLhotakECOOP13} is
a tool for generalizing call-graphs of applications by leaning on a
\emph{separate compilation assumption} to generate a general stub
library for applications. This allows analysis tools to be modular,
as generating full call-graphs 
is both expensive and imprecise.
They show how encapsulation assumptions and proofs can be leveraged to improve
the feasibility and the precision of analyses.
In our work, we give
a sufficient condition, ECF, for the ability to soundly reason about
a single object in isolation from any other object.

The work of~\citet{RW:Leino2002ValidStateBit} presents an idiom for
verifying if an object behaves as expected in the presence of
callbacks, called \emph{Valid/state specification idiom}. Every
object $o$ maintains a `valid' bit that indicates if its state is
valid, i.e., satisfies its object invariants. The bit should be true
in every first invocation of $o$ in an execution. When
$o$ calls a method of another object $o'$, $o$ turns off its `valid'
bit. This way, if the execution of $o'$ leads to another method call
of $o$, before the original call to $o$ completed, the code of $o$
can take into account the fact that its object invariants do not
necessarily hold. The existence of such a
`valid' bit is helpful to achieve modular soundness, that is the
ability to reason about an object in isolation.
This paper achieves modular soundness by relying on the encapsulation of the object's
state.
 Essentially, an ECF object is an object for which the
`valid' bit is always turned on, as it is guaranteed that the object
state changes only from within the object's methods, and that those
methods too are only executed where originally the `valid' bit would
be turned on. Thus, with the assumption on all executions being ECF,
there is no need to define a separate behavior of the code for when
the `valid' bit is turned on or off. To enable sound modular
reasoning, we simply ignore external calls and assume any return
value returned from any such external call.
We note that the absence of shared state drastically simplify our life.

\citet{LogozzoCLSS09}~presents a method for \emph{modular inference of class invariants}.
Specifically, it is shown that the trace semantics of an isolated class are sound and complete with respect to the trace semantics of a whole program.
The goal is to find the strongest state-based sound class invariant, that holds in both the isolated and non isolated cases.
Abstraction is used in order to compute such an invariant.
If it the class invariant matches the specification of the class,
 then it is ensured that the class itself matches the specification
 even in the context of a whole program.
The mentioned work enables modular reasoning by using abstraction.
Our work does not attempt to find such a sound class invariant,
but rather to satisfy the necessary conditions for being able to statically verify
any specification of an object in isolation of other objects.
The benefit here is that we do not depend on the precision of an abstraction,
which may output an invariant that overapproximates the specification, and thus does not meet it.


\subsection{Verification of Smart Contracts}
Even before the events surrounding the bug in the DAO,
there were discussions in the Ethereum community
 about formal verification of smart contracts.
Following the extreme measures taken to
avert the effects of the attack on the DAO by hard-forking the blockchain and effectively
rewrite its history of executions, the discussion became more wide-spread.

\citet{RW:LuuCCS16}~characterized a class of security bugs in smart contracts called \emph{Transaction-Ordering Dependence (TOD)}. A contract inflicted with TOD bugs may behave unexpectedly when there is more than one client using the system and the effect of the execution of one client depends on whether the other client already executed or not.
In both TOD and ECF the bugs arise from the fact that the execution is performed in an unexpected state of the contract.
However, TOD bugs arise when there is more than one execution (since smart contracts are executed in a distributed environment), whilst non-ECF arises even in a single execution which contains callbacks.
One of the solutions suggested for TOD bugs is \emph{guarded transactions}.
The idea is to allow contract writers to define guard conditions which are verified by the virtual machine executing the contract code.
The execution must satisfy the guard condition, otherwise it fails without any effect.
However, by enabling modular reasoning on contracts by proving or asserting at runtime that the executions are ECF, we can verify similar conditions statically.
The only addition that may be required for the virtual machine is the online ECF check, which we found to be inexpensive in practice.
Checking arbitrary conditions at runtime may either be inefficient,
or not expressible enough to specify fully correct contract behavior.
In addition, by verifying at runtime that executions are ECF,
 we are already able to detect and prevent executions which are,
 with high probability, unwanted or unexpected.

\citet{RW:LuuCCS16} presents a tool called \emph{Oyente}~\cite{Eth:OyenteOnline}, based on symbolic execution of contracts.
The tool's web version reports on the existence of `reentrancy bugs', which is how the family of bugs such as the bug in the DAO were dubbed by the Ethereum community.
We attempted to verify both ECF and non-ECF contract variations based on the DAO object presented in \Cref{Fi:DaoContract}.
We received a report on a `reentrancy bug', even on ECF contracts.
We reported the false positives to the web Oyente team, and will submit the issue request by the camera-ready deadline.

The Why3~\cite{RW:Why3filliatre13esop} tool was also applied to verify smart contracts written in Solidity.
This requires whole code analysis and user supplied loop invariants.

\emph{}~\citet{BhargavanPLAS16}~translate a subset of the high-level Solidity language for Smart Contract development to F*, enabling using F*'s verification framework on Smart Contracts.
They also presented a decompiler for EVM bytecode to F*.
Similarly to the Why3 approach, the authors faced the issue of translating peculiar syntactic features
of the smart contract language Solidity to F*.

It should be noted that both F* and Oyente are successful in detecting other bugs, such as mishandled exceptions.
For technical clarity, we omit discussion of the semantics of exceptions and rollbacks in Ethereum.
Primarily, to arrive at general results that can be applied in domains other than Ethereum,
and secondly, to not overbear the reader with technical details
on the myriad ways Ethereum contracts may be invoked, and how exceptions may be handled in each of these ways.

\citet{MillerIACR2015}~discuss their insights from an educational smart contracts lab they held, and
published example contracts used in the lab.
We manually analyzed one such contract, implementing a \emph{rock, paper, scissors} game~\cite{Eth:EthereumLab}.
We identified several control paths in which a non-ECF execution might manifest.
Specifically, there are two control paths in registration to the game (in which players provide a sum as bounty), and three additional paths in the collection of the prize.
However, the authors put a constraint on the ability to execute callbacks by limiting to a minimum the amount of \emph{gas} available to the execution.
\emph{gas} is a novel concept in Ethereum that effectively bounds the runtime by associating with each low-level opcode a cost.
If an execution is not provided with enough \emph{gas} when called, it throws a special \emph{out-of-gas} exception.

\citet{RW:SergeyArxiv17} offer an analogy between the nomenclature of Smart Contracts and that of concurrent objects.
Specifically, the scenario of a contract calling another contract is compared to cooperative multitasking,
in which contract invocation is analogous to the case where the caller yields control.
One of the main challenges mentioned is that of being able to verify contracts in isolation of other contracts.
The ECF property brings us closer to that goal,
 by allowing to check properties that can be specified as `contract invariants' in a modular way.

\section{Conclusion}\label{Sec:Conclusion}

In this paper we have presented a simple generic correctness
condition for callbacks called Effective Callback Freedom and
studied its usefulness.
We have shown that it enables modular reasoning in environments with local-only mutable states like Ethereum.
We have also shown that in Ethereum it can be used to prevent bugs without drastically limiting programming style, 
and that it can be checked dynamically with low runtime overhead.
In the future, we expect to apply the concept of ECF and prove its usefulness in other environments such as Microservices and Amazon~$\lambda$.




\begin{acks}  
We would like to thank the reviewers for their helpful comments. The research leading to these results has received funding from the  European Research Council under the European Union's Seventh Framework Programme (FP7/2007-2013) / ERC grant agreement n$^{\circ}$ [321174], Len Blavatnik and the Blavatnik Family foundation, Blavatnik Interdisciplinary Cyber Research Center at Tel Aviv University, and the Pazy Foundation.
\end{acks}



\bibliography{biblio}


\newpage

\appendix

\end{document}